\documentclass[11pt,letterpaper]{article}

\usepackage[margin=1in]{geometry}
\usepackage[T1]{fontenc}
\usepackage[utf8]{inputenc}
\usepackage{amsmath,amssymb,amsthm}
\usepackage{bm}
\usepackage{thm-restate}
\usepackage{graphicx}
\usepackage{booktabs}
\usepackage{flafter}
\usepackage{xcolor}
\usepackage{tikz}
\usetikzlibrary{calc,positioning,fit,arrows.meta,shapes}
\newcommand{\ThreeColoringConstruction}[1]{%
  \begingroup
  \definecolor{rowblue}{HTML}{A5CFE6}
  \definecolor{roworange}{HTML}{F0C286}
  \begin{tikzpicture}[
    x=1cm,y=1cm,font=\scriptsize,
    vertex/.style={circle,draw=black!85,line width=.5pt,
      minimum size=14pt,inner sep=.4pt,outer sep=0pt},
    graph edge/.style={draw=black!80,line width=.55pt,line cap=round},
    ablock/.style={draw=black!25,fill=white,line width=.3pt,rounded corners=1.5pt},
    bblock/.style={draw=black!25,fill=black!3,line width=.3pt,rounded corners=1.5pt}
  ]
  \ifnum#1=1\relax
    \tikzset{col0/.style={fill=black!16},col1/.style={fill=rowblue},col2/.style={fill=roworange}}
  \else
    \tikzset{col0/.style={fill=white},col1/.style={fill=white},col2/.style={fill=white}}
  \fi

  \foreach \kind/\idx/\leftx/\rightx/\blockstyle in {
    A/1/0.000/1.350/ablock,
    B/1/1.410/2.210/bblock,
    A/2/2.270/3.620/ablock,
    B/2/3.680/4.480/bblock,
    A/3/4.540/5.890/ablock,
    B/3/5.950/6.750/bblock,
    A/4/6.810/9.310/ablock,
    B/4/9.370/10.170/bblock,
    A/5/10.230/11.580/ablock,
    B/5/11.640/12.440/bblock,
    A/6/12.500/13.850/ablock,
    B/6/13.910/14.710/bblock,
    A/7/14.770/16.120/ablock%
  } {
    \path[\blockstyle] (\leftx,-.45) rectangle (\rightx,5.25);
    \node at ({(\leftx+\rightx)/2},4.97) {$\kind_{i,\idx}$};
  }

  \def\rowvertices{
    alpha1/{\alpha_{1}}/0.675/4.300/0,
    beta1/{\beta_{1}}/1.085/3.350/1,
    alpha2/{\alpha_{2}}/2.945/4.300/0,
    beta2/{\beta_{2}}/3.355/3.350/1,
    gamma2/{\gamma_{2}}/2.535/3.350/1,
    alpha3/{\alpha_{3}}/5.215/4.300/0,
    beta3/{\beta_{3}}/5.625/3.350/2,
    gamma3/{\gamma_{3}}/4.805/3.350/1,
    alpha4/{\alpha_{4}}/8.060/4.300/0,
    beta4/{\beta_{4}}/9.010/3.350/1,
    gamma4/{\gamma_{4}}/7.110/3.350/2,
    alpha5/{\alpha_{5}}/10.905/4.300/0,
    beta5/{\beta_{5}}/11.315/3.350/2,
    gamma5/{\gamma_{5}}/10.495/3.350/1,
    alpha6/{\alpha_{6}}/13.175/4.300/0,
    beta6/{\beta_{6}}/13.585/3.350/2,
    gamma6/{\gamma_{6}}/12.765/3.350/2,
    alpha7/{\alpha_{7}}/15.445/4.300/0,
    gamma7/{\gamma_{7}}/15.035/3.350/2,
    h1/{h_{1}}/1.810/6.000/1,
    mu1/{\mu_{1}}/1.810/4.300/2,
    w1/{w_{1}}/1.810/3.650/0,
    z1/{z_{1}}/1.810/2.950/2,
    h2/{h_{2}}/4.080/6.000/1,
    mu2/{\mu_{2}}/4.080/4.300/2,
    w2/{w_{2}}/4.080/3.650/0,
    z2/{z_{2}}/4.080/2.950/2,
    h3/{h_{3}}/6.350/6.000/2,
    mu3/{\mu_{3}}/6.350/4.300/1,
    w3/{w_{3}}/6.350/3.650/0,
    z3/{z_{3}}/6.350/2.950/1,
    h4/{h_{4}}/9.770/6.000/1,
    mu4/{\mu_{4}}/9.770/4.300/2,
    w4/{w_{4}}/9.770/3.650/0,
    z4/{z_{4}}/9.770/2.950/2,
    h5/{h_{5}}/12.040/6.000/2,
    mu5/{\mu_{5}}/12.040/4.300/1,
    w5/{w_{5}}/12.040/3.650/0,
    z5/{z_{5}}/12.040/2.950/1,
    h6/{h_{6}}/14.310/6.000/2,
    mu6/{\mu_{6}}/14.310/4.300/1,
    w6/{w_{6}}/14.310/3.650/0,
    z6/{z_{6}}/14.310/2.950/1,
    lambda2/{\lambda_{2}}/2.945/2.300/2,
    lambda4/{\lambda_{4}}/8.060/2.300/2,
    lambda5/{\lambda_{5}}/13.175/2.300/1,
    p13/{p^{1}_{3}}/4.915/1.800/2,
    q13/{q^{1}_{3}}/5.515/1.100/1,
    p14/{p^{1}_{4}}/7.110/1.800/2,
    q14/{q^{1}_{4}}/7.610/1.100/1,
    p24/{p^{2}_{4}}/9.010/1.800/1,
    q24/{q^{2}_{4}}/8.510/1.100/2,
    p25/{p^{2}_{5}}/10.605/1.800/1,
    q25/{q^{2}_{5}}/11.205/1.100/2,
    x12/{x^{1}_{2}}/4.080/1.800/0,
    y12/{y^{1}_{2}}/4.080/1.100/2,
    x13/{x^{1}_{3}}/6.350/1.800/1,
    y13/{y^{1}_{3}}/6.350/1.100/0,
    x24/{x^{2}_{4}}/9.770/1.800/2,
    y24/{y^{2}_{4}}/9.770/1.100/0,
    x25/{x^{2}_{5}}/12.040/1.800/0,
    y25/{y^{2}_{5}}/12.040/1.100/1,
    s1/{s^{1}}/8.060/0.300/1,
    s2/{s^{2}}/7.360/-0.150/0,
    s3/{s^{3}}/8.760/-0.150/2%
  }
  \foreach \name/\symbol/\vx/\vy/\col in \rowvertices
    \coordinate (\name) at (\vx,\vy);

  \foreach \idx in {1,...,6}
    \draw[draw=black!55,line width=.55pt] (h\idx) -- ++(0,-.75);

  \foreach \from/\to in {
    alpha1/beta1, alpha1/mu1, alpha2/beta2, alpha2/gamma2,
    alpha2/lambda2, alpha2/mu1, alpha2/mu2, alpha3/beta3,
    alpha3/gamma3, alpha3/mu2, alpha3/mu3, alpha3/p13,
    alpha3/q13, alpha4/beta4, alpha4/gamma4, alpha4/lambda4,
    alpha4/mu3, alpha4/mu4, alpha4/p14, alpha4/p24,
    alpha4/q14, alpha4/q24, alpha5/beta5, alpha5/gamma5,
    alpha5/mu4, alpha5/mu5, alpha5/p25, alpha5/q25,
    alpha6/beta6, alpha6/gamma6, alpha6/lambda5, alpha6/mu5,
    alpha6/mu6, alpha7/gamma7, alpha7/mu6, beta1/w1,
    beta1/z1, beta2/lambda2, beta2/w2, beta2/z2,
    beta3/w3, beta3/z3, beta4/lambda4, beta4/w4,
    beta4/z4, beta5/w5, beta5/z5, beta6/w6,
    beta6/z6, gamma2/w1, gamma2/z1, gamma3/w2,
    gamma3/z2, gamma4/w3, gamma4/z3, gamma5/w4,
    gamma5/z4, gamma6/lambda5, gamma6/w5, gamma6/z5,
    gamma7/w6, gamma7/z6, lambda2/x12, lambda4/s1,
    lambda5/x25, p13/q13, p13/x12, p13/x13,
    p14/q14, p14/s2, p14/x13, p24/q24,
    p24/s3, p24/x24, p25/q25, p25/x24,
    p25/x25, q13/y12, q13/y13, q14/y13,
    q24/y24, q25/y24, q25/y25, s1/s2,
    s1/s3, s2/s3, w1/z1, w2/z2,
    w3/z3, w4/z4, w5/z5, w6/z6,
    x12/y12, x13/y13, x24/y24, x25/y25%
  } \draw[graph edge] (\from) -- (\to);

  \foreach \name/\symbol/\vx/\vy/\col in \rowvertices
    \node[vertex,col\col] at (\name) {$\symbol$};
  \end{tikzpicture}%
  \endgroup
}
\usepackage{enumitem}
\usepackage[hidelinks]{hyperref}
\usepackage[nameinlink,capitalize,noabbrev]{cleveref}
\usepackage{xspace}

\newtheorem{theorem}{Theorem}[section]
\newtheorem{lemma}[theorem]{Lemma}
\newtheorem{question}[theorem]{Question}
\newtheorem{proposition}[theorem]{Proposition}

\newtheorem{claim}[theorem]{Claim}
\newtheorem{conjecture}[theorem]{Conjecture}
\theoremstyle{definition}

\theoremstyle{remark}

\providecommand{\theHtheorem}{}
\renewcommand{\theHtheorem}{\thesection.\arabic{theorem}}
\makeatletter
\@for\@tempa:=lemma,question,proposition,corollary,claim,conjecture,definition,remark\do{%
  \expandafter\def\csname theH\@tempa\endcsname{\theHtheorem}%
}
\makeatother

\crefname{conjecture}{Conjecture}{Conjectures}
\Crefname{conjecture}{Conjecture}{Conjectures}
\crefname{question}{Question}{Questions}
\Crefname{question}{Question}{Questions}
\crefname{claim}{Claim}{Claims}
\Crefname{claim}{Claim}{Claims}

\newcommand{\Oh}{\mathcal{O}}
\newcommand{\nae}{\text{NAE}}
\newcommand{\mis}{\textsc{Max Independent Set}\xspace}
\newcommand{\smis}{\textsc{MIS}\xspace}
\newcommand{\msot}{\ensuremath{\mathrm{MSO}_2}\xspace}
\newcommand{\cmsot}{\ensuremath{\mathrm{CMSO}_2}\xspace}
\DeclareMathOperator{\maxcut}{MaxCut}

\DeclareMathOperator{\tw}{tw}
\DeclareMathOperator{\tww}{tww}

\DeclareMathOperator{\compmaxleaf}{cml}
\newcommand{\reduced}[1]{{#1}^{\downarrow}}

\renewcommand{\le}{\leqslant}
\renewcommand{\ge}{\geqslant}
\renewcommand{\leq}{\leqslant}
\renewcommand{\geq}{\geqslant}

\makeatletter
\newcommand{\proofparagraph}{%
  \@startsection{paragraph}{4}{\z@}%
    {\medskipamount}{-1em}{\normalfont\normalsize\bfseries}}
\makeatother

\title{Max Independent Set Remains NP-hard when Excluding a Planar Induced Minor (and More)}
\author{%
  \'{E}douard Bonnet\thanks{CNRS, ENS de Lyon, Université Claude Bernard Lyon 1, LIP UMR 5668, Lyon, France.}
  \and
  Yeonsu Chang\thanks{Department of Mathematics, Hanyang University, Seoul, South Korea.}%
}
\date{}

\begin{document}

\maketitle

\begin{abstract}
We show that \mis (\smis) remains NP-hard in graphs excluding the $5 \times 5$ grid as an induced minor.
This negatively answers a~question of Dallard--Milani\v{c}--\v{S}torgel that has attracted sustained attention, and refutes a~conjecture by Gartland and Lokshtanov.

Unless $N$-variable \textsc{3-SAT} can be solved in time $2^{\Oh(N^{1-\varepsilon})}$ for some $\varepsilon > 0$, \smis requires $2^{\Omega(n^{1/2-o(1)})}$ time on $n$-vertex graphs in this class.
This complements a~$2^{\widetilde{\Oh}(n^{2/3})}$-time algorithm on any $H$-induced-minor-free class by Korhonen and Lokshtanov.
Exploiting a~connection between low twin-width or reduced parameters on one hand and planar-induced-minor-freeness on the other, we further show that \textsc{$3$-Coloring}, \textsc{Feedback Vertex Set} and \textsc{Dominating Set} are NP-complete on weakly sparse graphs of reduced component max-leaf (hence twin-width) at most 3, even when a witnessing contraction sequence is given.
Classes of bounded reduced component max-leaf (by any constant) and of~twin-width at most 3 are known to exclude a~fixed planar graph as an induced minor.
The \textsc{$3$-Coloring} and \textsc{Feedback Vertex Set} results refute more cases of the Gartland--Lokshtanov conjecture, even when the inputs are weakly sparse.
The complexity of \smis in weakly sparse graphs excluding a~planar induced minor remains open. 
\end{abstract}

\section{Introduction}\label{sec:intro}

Graph classes excluding a~planar minor have bounded treewidth~\cite{RobertsonS86}.
Thus, in them, any problem definable in \msot logic (or its optimization variants) can be solved in polynomial time~\cite{Courcelle90}.
This cannot hold in classes excluding a~planar \emph{induced} minor (where induced minors are obtained by removing vertices and contracting edges).
Indeed, a~problem like \textsc{Dominating Set} is NP-complete on graphs excluding a~5-vertex path as an induced minor~\cite{Bertossi84} (or, equivalently, as an induced subgraph).

However, no similar direct argument dismisses that possibility for the \mis problem (\smis for short).
This led Dallard, Milanič, and Štorgel to ask whether a~polynomial-time algorithm exists for \smis on any planar-induced-minor-free class.

\begin{question}[Dallard--Milanič--Štorgel~\cite{DallardMS24a}]\label{conj:dms}
Is it true that, for every planar graph $H$, \mis can be solved in polynomial time on graphs excluding $H$ as an induced minor?
\end{question}

This question has garnered a~large amount of related work in the past five years.
Let $K_t$ denote the $t$-vertex clique, $K_{s,t}$ the biclique with $s$ vertices against $t$ vertices, $P_t$ the $t$-vertex path, $C_t$ the $t$-vertex cycle, and $W_t$ the wheel obtained from $C_t$ by adding a~universal vertex.
\Cref{conj:dms} was positively solved when $H$ is $K_5$ minus an edge, $K_{2,t}$, $W_4$~\cite{DallardMS24a}, $P_6$~\cite{Grzesik22}, or any windmill graph (a~matching plus a~universal vertex)~\cite{BonnetDGTW26}.

A~natural relaxation of~\cref{conj:dms} to quasipolynomial time has seen further development.

\begin{conjecture}[Gartland--Lokshtanov, Korhonen~\cite{Korhonen23}, {\cite[Question 2]{BonnetDGTW26}}]\label{conj:qdms}
For every planar graph~$H$, \mis can be solved in quasipolynomial time on graphs excluding $H$ as an induced minor.
\end{conjecture}

While \cref{conj:dms} is already open when $H = P_7$, \cref{conj:qdms} is known when $H$ is any path, any cycle~\cite{Gartland21}, a~disjoint union of triangles~\cite{BonamyBDEGHTW24}, or even~$tC_3~\uplus~C_4$~\cite{BonnetDGTW26}.

Korhonen and Lokshtanov have obtained a~subexponential algorithm, as a~byproduct of balanced separators of size $\Oh(\sqrt m)$, where $m$ is the number of edges.

\begin{theorem}[{\cite[Corollary 1.2]{KL24-InducedMinor}}]
  For every planar graph $H$, \mis can be solved in $2^{\Tilde{\Oh}_H(n^{2/3})}$ time on $n$-vertex graphs excluding $H$ as an induced minor.
\end{theorem}

Gartland and Lokshtanov have proposed strengthenings of~\cref{conj:dms,conj:qdms} for a~family of problems.
Fix a~positive integer $t$ and a~\cmsot sentence $\varphi$, i.e., a~closed formula with universal and existential quantifications over vertices, vertex subsets, and edge subsets, and counting predicates modulo some fixed integers.
The $(\tw \leqslant t,\varphi)$-\textsc{Weighted Maximum Induced Subgraph} ($(\tw \leqslant t,\varphi)$-\textsc{WMIS}) takes as input a~vertex-weighted graph $(G,w)$ and asks one either to find a~set $S \subseteq V(G)$ such that
\begin{itemize}
\item $G[S]$ satisfies $\varphi$,
\item the treewidth of $G[S]$ is at~most $t$, and
\item the sum of weights of vertices in~$S$ is maximized subject to the above conditions,
\end{itemize}
or to conclude that no such set $S$ exists.

Special cases of the $(\tw \leqslant t,\varphi)$-\textsc{WMIS} problem include, for instance, \smis, \textsc{Feedback Vertex Set}, and \textsc{Even Cycle Transversal}.
It was conjectured that a~positive answer to~\cref{conj:dms} more generally holds for $(\tw \leqslant t,\varphi)$-\textsc{WMIS} and \textsc{$3$-Coloring}.

\begin{conjecture}[{\cite[Conjecture 1.4.2]{G23-Thesis}}]\label{conj:gl}
  For every planar graph $H$, positive integer $t$, and \cmsot sentence $\varphi$, $(\tw \leqslant t,\varphi)$-\textsc{Weighted Maximum Induced Subgraph} and \textsc{$3$-Coloring} can be solved in polynomial time on $H$-induced-minor-free graphs.
\end{conjecture}

\subparagraph*{Our contributions.}
We negatively answer \cref{conj:dms} and refute \cref{conj:qdms,conj:gl} (unless P $=$ NP for \cref{conj:dms,conj:gl}, and unless NP $\subseteq$ QP for \cref{conj:qdms}) by showing the following.

\begin{restatable}{theorem}{misHardness}\label{thm:main}\label{thm:mis}
\smis is NP-hard in graphs excluding the $5 \times 5$ grid as an induced minor.
\end{restatable}

The proof is a~remarkably simple reduction from \textsc{Max Cut} on general graphs~$F$.

Look at \cref{app:mis:fig:intermediate-graph}: this intermediate graph $G$ has the strong product of a~path with an edge (i.e., a~ladder with diagonals) for each \emph{row} (with each path occupying one of two sides), and a~biclique between the two sides in each column.
It is easy to show that a~cycle in~$G$ contains two vertices in the same column and on opposite sides (mentally remove all the column bicliques and consider the possible neighbors of a~leftmost vertex of the cycle).
The union of the neighborhoods of two such vertices contains the entire biclique and thus separates the left-hand side from the right-hand side.
Therefore, every graph with three vertex-disjoint nonadjacent cycles and a~path between every pair of cycles avoiding the neighborhood of the third cycle cannot be an induced minor of~$G$ (the $5 \times 5$ grid has this property; see the right side of~\cref{app:mis:fig:grid-three-cycles}).
Indeed, consider the cycle whose neighborhood contains the biclique in between the bicliques of the other two cycles.

We then pick an appropriate induced subgraph $G_F$ of~$G$, where every odd-index column represents a~vertex of~$F$, and every edge of~$F$ ``occupies'' two rows.
From the rows of edge $e=uv \in E(F)$, we only keep the vertices in columns between that of~$u$ and that of~$v$, with the exception of four vertices; see~\cref{app:mis:fig:edge-gadget}.
This defines the graph $G_F$.
Any independent set $I$ in~$G_F$ defines a~bipartition of~$V(F)$ where the part of $v \in V(F)$ depends on the side of the column of~$v$ that intersects~$I$.
This is well-defined due to the column bicliques.
Finally, the encoding of $e \in E(F)$ is such that one more vertex can be inserted in the independent set exactly when its endpoints are in distinct parts of the bipartition of~$V(F)$; see~\cref{app:mis:fig:alpha-lower-bound}.

If we reduce from \textsc{Max Cut} on (sub)cubic graphs~$F$, the previous reduction produces graphs with $\Oh(|V(F)|^2)$ vertices.
We call \emph{Weak Exponential-Time Hypothesis} (WETH) the assumption that, for every fixed $\varepsilon > 0$, $N$-variable \textsc{3-SAT} cannot be solved in time $2^{\Oh(N^{1-\varepsilon})}$.
Together with the Sparsification Lemma~\cite{sparsification}, our reduction yields the following.

\begin{restatable}{corollary}{misLowerBound}\label{thm:eth}\label{mis:eth}
Unless WETH fails, for every fixed $\varepsilon > 0$, there is no $2^{\Oh(n^{1/2-\varepsilon})}$-time algorithm for \mis on $n$-vertex graphs excluding the $5 \times 5$ grid as an induced minor.
\end{restatable}

The produced \smis instances have large bicliques.
We leave as an open question whether \smis admits a~(quasi)polynomial-time algorithm on \emph{weakly sparse} classes (i.e., excluding a~fixed biclique $K_{t,t}$ as a~subgraph) excluding a~planar induced minor.
However, we rule that out for two problems covered by~\cref{conj:gl}: \textsc{$3$-Coloring} and \textsc{Feedback Vertex Set}.

For that, we use a~connection between graph parameters obtained via contraction sequences (in particular, twin-width and reduced component max-leaf) and planar-induced-minor-freeness.
The definition of contraction sequences and of these parameters can be found in~\cref{subsec:contraction}.
It was shown that every class of bounded reduced component max-leaf excludes a~fixed grid as an induced minor~\cite{BCDGK26-RCML}, and so does the class of graphs with twin-width at~most~3~\cite{ACHO25-Subdivisions}.
The latter was explicitly mentioned in~\cite{B26-Coloring} as a~possible avenue to refute conjectures on classes excluding a~planar induced minor.

The twin-width of a~graph $\tww$ is at~most its reduced component max-leaf $\reduced{\compmaxleaf}$.
We show that several graph problems remain NP-complete on weakly sparse graphs of reduced component max-leaf at~most~3.
This is in contrast with polynomial-time algorithms for \smis and \textsc{Max Induced Matching} in graphs of bounded reduced component max-leaf, when a~witnessing contraction sequence is given as part of the input~\cite{BCDGK26-RCML}.

\begin{restatable}{theorem}{coloringHardness}\label{col:theorem}
\textnormal{\textsc{$3$-Coloring}} is NP-complete on a~weakly sparse class with reduced component max-leaf (hence twin-width) at~most~3, even with a witnessing contraction sequence.
\end{restatable}

\begin{restatable}{theorem}{fvsHardness}\label{fvs:hardness}
\textnormal{\textsc{Feedback Vertex Set}} is NP-complete on a~weakly sparse class with reduced component max-leaf (hence twin-width) at~most~3, even with a witnessing contraction sequence.
\end{restatable}

\begin{restatable}{theorem}{dsHardness}\label{ds:hardness}
\textnormal{\textsc{Dominating Set}} is NP-complete on a~weakly sparse class with reduced component max-leaf (hence twin-width) at~most~3, even with a witnessing contraction sequence.
\end{restatable}

By~\cref{prop:planar-exclusion}, \cref{col:theorem,fvs:hardness} imply the following further refutations of~\cref{conj:gl} even within weakly sparse classes, unless P $=$ NP.

\begin{restatable}{corollary}{planarHardness}\label{cor:planar-hardness}
There is a fixed planar graph $H_3$ such that \textnormal{\textsc{$3$-Coloring}} and \textnormal{\textsc{Feedback Vertex Set}} are NP-complete on $K_{3,3}$-subgraph-free $H_3$-induced-minor-free classes.
\end{restatable}

Besides refuting~\cref{conj:gl}, \cref{col:theorem,fvs:hardness,ds:hardness} clarify the tractability limits in classes of bounded reduced component max-leaf, and in classes of low twin-width.
It is observed in~\cite{B26-Coloring} that most NP-complete graph problems remain hard in graphs of twin-width at~most~4, but are polynomial-time solvable in graphs of twin-width at~most~1 (which is a~class of bounded clique-width).
Previously known results for twin-width at~most~2 and at~most~3 include a~polynomial-time algorithm for \smis in graphs given with a~2-sequence (i.e., a~witness that the twin-width is at~most~2)~\cite{BGKTW24-TwwIII}, and the NP-hardness of \textsc{Chromatic Number} on graphs of twin-width at~most~3~\cite{B26-Coloring}.
\cref{tab:ipl-update} summarizes these results and our new hardness results.

\begin{table}[!ht]
\centering
\small
\setlength{\tabcolsep}{3.5pt}
\begin{tabular}{lccccc}
& \multicolumn{4}{c}{\footnotesize planar-induced-minor-free} & \\
\cmidrule(lr){2-5}
& {\footnotesize bounded clique-width} & & & & \\
\cmidrule(lr){2-2}
Problem & $\mathcal{T}_1$ & $\mathcal{T}_2$ & $\reduced{\compmaxleaf}\le3$ & $\mathcal{T}_3$ & $\mathcal{T}_4$ \\
\midrule
\textsc{Max Independent Set} & $\mathrm P$~{\footnotesize\cite{BKRT22-TwwVI,CourcelleMR00}} & $\mathrm P^\ast$~{\footnotesize\cite{BGKTW24-TwwIII}} & $\mathrm P^\ast$~{\footnotesize\cite{BCDGK26-RCML}} & ? & NP-c~{\footnotesize\cite{BBD22-Tww4}} \\
\textsc{$3$-Coloring} & $\mathrm P$~{\footnotesize\cite{BKRT22-TwwVI,CourcelleMR00}} & ? & \textbf{NP-c} & \textbf{NP-c} & NP-c~{\footnotesize\cite{B26-Coloring}} \\
\textsc{Feedback Vertex Set} & $\mathrm P$~{\footnotesize\cite{BKRT22-TwwVI,CourcelleMR00}} & ? & \textbf{NP-c} & \textbf{NP-c} & NP-c~{\footnotesize\cite{BBD22-Tww4}} \\
\textsc{Dominating Set} & $\mathrm P$~{\footnotesize\cite{BKRT22-TwwVI,CourcelleMR00}} & ? & \textbf{NP-c} & \textbf{NP-c} & NP-c~{\footnotesize\cite{BBD22-Tww4}} \\
\textsc{Chromatic Number} & $\mathrm P$~{\footnotesize\cite{AJKPRW25-TwwOne,EPL72-Permutation}} & ? & \textbf{NP-c} & NP-c~{\footnotesize\cite{B26-Coloring}} & NP-c~{\footnotesize\cite{B26-Coloring}} \\
\bottomrule
\end{tabular}
\caption{Complexity of the decision versions, with $\mathcal{T}_d=\{G:\tww(G)\le d\}$.
Entries in bold follow from this paper.
Here, $\mathrm P^\ast$ means that a witnessing contraction sequence is required by the algorithm.}
\label{tab:ipl-update}
\end{table}

For every graph $G$, we have $\tww(G) \le 2$ if and only if $\reduced{\compmaxleaf}(G) \le 2$: red graphs of maximum degree at most $2$ are disjoint unions of paths and cycles.
It remains to understand whether graphs of twin-width at~most~2 allow polynomial algorithms for the problems of~\cref{tab:ipl-update}.
The complexity of \smis on graphs of twin-width at~most~3 remains open, too~\cite[Question~5]{B26-Coloring}.

\subparagraph*{Future work.}
Now, \cref{conj:dms,conj:qdms} can be turned into dichotomy questions.

\begin{question}\label{q:dichotomy}
  Which planar graphs $H$ are such that \mis can be solved in (quasi)polynomial time on graphs excluding $H$ as an induced minor?
\end{question}

In fact, we give a~negative answer to~\cref{conj:dms,conj:qdms} for a~graph simpler than the $5 \times 5$ grid: the cycle $C_9$ (with vertices $1, \ldots, 9$ in cyclic order) with three additional vertices respectively adjacent to 1~and 2, 4 and 5, and 7 and 8.

One can also turn to approximation schemes.

\begin{conjecture}\label{conj:qptas}
  For every planar graph~$H$, \mis admits a~quasipolynomial-time approximation scheme (QPTAS) on graphs excluding $H$ as an induced minor.
\end{conjecture}

The authors recently proved \cref{conj:qptas} in greater generality~\cite{BC26-QPTAS}.
They showed that \textsc{Max Weight Independent Set} admits a~deterministic QPTAS on every hereditary class whose $m$-edge graphs have balanced separators of size $\Oh(m^b)$ for some fixed $0\le b<1$.
In particular, this applies to $H$-induced-minor-free graphs for every fixed graph~$H$.

The following structural conjecture of~Gartland and Lokshtanov, which implies~\cref{conj:qptas}, remains open.

\begin{conjecture}[Gartland--Lokshtanov~\cite{G23-Thesis}]\label{conj:struct-gl}
For every planar graph $H$, there is a~constant $k := k(H)$ such that every $n$-vertex $H$-induced-minor-free graph $G$ admits a~set $X \subseteq V(G)$ of size at~most~$k$ such that the closed neighborhood of~$X$ is a~balanced separator of $G$.
\end{conjecture}

Whether \smis admits a~polynomial-time approximation scheme (PTAS) on every class excluding a~fixed (planar) induced minor remains open.

\paragraph{Organization.}
\cref{sec:preliminaries} introduces reduced component max-leaf and twin-width, and gives their connection with planar-induced-minor-freeness.
\cref{sec:mis} presents the MIS reduction.
\cref{sec:further-hardness} treats \mbox{\textsc{$3$-Coloring}} and \textsc{Feedback Vertex Set}, and \cref{sec:domination} gives the \textsc{Dominating Set} reduction.
\cref{sec:sparsity} proves that the reductions of~\cref{sec:further-hardness,sec:domination} produce $K_{3,3}$-subgraph-free graphs.

\section{Preliminaries}\label{sec:preliminaries}

For integers $i,j$, set $[i,j]=\{k\in\mathbb Z:i\le k\le j\}$ and $[i]=[1,i]$.

\subsection{Standard graph notation and definitions} 

All graphs are finite and simple.
We write $V(G)$ and $E(G)$ for the vertex and edge sets, $G[S]$ for the subgraph induced by $S$ and $N_G[S]$ for the closed neighborhood of $S$.
The independence number $\alpha(G)$ is the size of a largest independent set.
We write $\maxcut(G)$ for the number of edges in a maximum-cardinality cut.

An \emph{induced minor} of a~graph $G$ is any graph obtained from $G$ by removing vertices and contracting edges.
An \emph{induced-minor model} of $H$ in $G$ consists of pairwise disjoint nonempty connected \emph{branch sets} $B_v$ for $v \in V(H)$ such that there is an edge between distinct branch sets $B_u$ and $B_v$ if and only if $uv \in E(H)$.
Naturally, $H$ is an induced minor of~$G$ if and only if there is an induced-minor model of~$H$ in~$G$.
If $H$ is not an induced minor of~$G$, we write that $G$ is $H$-induced-minor-free.
A~graph class $\mathcal C$ is (planar-)induced-minor-free if there is a~fixed (planar) graph $H$ such that all the graphs of~$\mathcal C$ are $H$-induced-minor-free. 

\subsection{Contraction sequences, twin-width, and reduced component max-leaf}\label{subsec:contraction}

A \emph{contraction sequence} of $G$ successively contracts two parts of a partition of $V(G)$, starting with singletons and ending with one part.
The contracted parts need not be adjacent.
Two parts are joined by a black edge if they are complete to one another, by no edge if they are anti-complete and by a red edge otherwise.
The \emph{red graph} consists of the parts and red edges.
The \emph{red neighbors} and \emph{red degree} of a part are its neighbors and degree in this graph, respectively.

\paragraph{Twin-width.}
The \emph{twin-width} of a graph $G$, denoted by $\tww(G)$, is the minimum integer $d$ such that $G$ admits a contraction sequence in which every red graph has maximum degree at most $d$.

\paragraph{Reduced component max-leaf.}
The \emph{component max-leaf} of a graph $G$, denoted by $\compmaxleaf(G)$, is the maximum number of leaves in any spanning tree of any component of $G$.
A one-vertex tree has one leaf.
The \emph{reduced component max-leaf} of $G$, denoted by $\reduced{\compmaxleaf}(G)$, is the minimum integer $k$ such that $G$ admits a contraction sequence in which every red graph has component max-leaf at most $k$~\cite{BCDGK26-RCML}.
Since every graph has maximum degree at most its component max-leaf, the same contraction sequence also witnesses $\tww(G)\le k$.
In particular, for every graph $G$, we have $\tww(G)\le\reduced{\compmaxleaf}(G)$.

\subsection{Planar-induced-minor-freeness}

The following structural results relate these width parameters to the exclusion of planar induced minors.

\begin{proposition}[{\cite{BCDGK26-RCML,ACHO25-Subdivisions,CampbellEtAl2024}}]
  \label{prop:planar-exclusion}
  There is a~map $f \colon \mathbb N \to \mathbb N$ such that every graph $G$ with $\reduced{\compmaxleaf}(G)\le r$ excludes the $f(r) \times f(r)$ grid as an induced minor.
There is also a fixed square grid $H_3$ such that every graph of twin-width at most $3$ is $H_3$-induced-minor-free.
\end{proposition}
\begin{proof}
The first statement follows from~\cite{BCDGK26-RCML}.
By~\cite[Corollary~1.9]{ACHO25-Subdivisions} with $H=K_5$, planar graphs of twin-width at most $3$ have treewidth at most some constant $w$.
By~\cite[Theorem~16]{CampbellEtAl2024}, choose a square grid $H_3$ such that every graph containing it as an induced minor has a planar induced subgraph of treewidth greater than $w$.
Since twin-width does not increase under taking induced subgraphs, no graph of twin-width at most $3$ contains $H_3$ as an induced minor.
\end{proof}

\paragraph{Consequences for~\cref{conj:gl}.}
The reductions in \cref{sec:further-hardness,sec:domination} produce graphs of reduced component max-leaf at most $3$ and hence twin-width at most $3$.
Therefore, these graphs exclude the same fixed square grid $H_3$ as an induced minor.

In~\cref{conj:gl}, take the treewidth bound $t=1$, let $\varphi$ be a tautological \cmsot{} sentence, and give every vertex weight $1$.
Then, $(\tw\le1,\varphi)$-\textsc{Weighted Maximum Induced Subgraph} asks for a largest set $S\subseteq V(G)$ inducing a forest, since the graphs of treewidth at most $1$ are exactly the forests.
The set $V(G)\setminus S$ is a minimum feedback vertex set of $G$.
Therefore, the hardness results in \cref{sec:further-hardness} refute the \textsc{$3$-Coloring} and \textsc{Feedback Vertex Set} cases of \cref{conj:gl} unless $\mathrm P=\mathrm{NP}$.

\section{Hardness of \textsc{Max Independent Set}}\label{sec:mis}\label{app:mis}

We reduce \textsc{Max Cut} in general graphs to \mis{} in graphs excluding the $5\times5$ grid as an induced minor.
We first construct an intermediate graph $G$ and show that it excludes this grid.
We then choose an induced subgraph $G_F$ whose independence number determines the maximum cut of $F$.

\misHardness*

\paragraph{The reduction.}
Let $F$ be any $n$-vertex $m$-edge \textsc{Max Cut} instance.
We arbitrarily order its vertices as $v_1,\ldots,v_n$.
We first build $G$ with vertex set
\[
 V(G)=\{(i,j,s):i\in[2m],\ j\in[2n-1],\ s\in\{0,1\}\}
\]
and edge set
\begin{align*}
 E(G)={}&\{(i,j,0)(i',j,1):i,i'\in[2m],\ j\in[2n-1]\}\\
 &\cup\{(i,j,s)(i,j',s'):i\in[2m],\ j,j'\in[2n-1],\
       |j-j'|=1,\ s,s'\in\{0,1\}\}.
\end{align*}
We call $i$ the row index, $j$ the column index and $s$ the side.
For each $j$, the column $C_j=\{(i,j,s):i\in[2m],\ s\in\{0,1\}\}$ induces a biclique between its two sides.
Along each row, all edges are present between consecutive columns.
The odd column $C_{2j-1}$ represents $v_j$; even columns do not represent any vertices of~$F$.
\cref{mis:fig:construction} shows $G$ and the two rows retained for one edge of $F$.

We reserve two rows $R'_{e,0},R'_{e,1}$ for every $e\in E(F)$ and denote their row indices by $e^0,e^1$.
For $e=v_jv_{j'}$ with $j<j'$ and each $s\in\{0,1\}$, we retain from row $e^s$ exactly the set
\begin{align*}
 R_{e,s}={}&\{(e^s,2j-1,s),(e^s,2j'-1,1-s)\}\\
 &\cup\{(e^s,\ell,t):2j-1<\ell<2j'-1,\ t\in\{0,1\}\}.
\end{align*}
We define
\[
 G_F=G\left[\bigcup_{e\in E(F)}(R_{e,0}\cup R_{e,1})\right],
 \qquad
 S_s=\{(i,j,s)\in V(G_F)\}\quad(s\in\{0,1\}).
\]
This completes the construction.
Both $G$ and $G_F$ have $\Oh(mn)$ vertices and can be constructed in polynomial time.

\begin{figure}[t]
\centering
\begin{tikzpicture}[
  edge/.style={draw=black, line width=.35pt, line cap=round},
  vertex/.style={circle, draw=black, line width=.45pt, minimum size=3.7pt, inner sep=0pt},
  side 0/.style={vertex, fill=white},
  side 1/.style={vertex, fill=black}
]
  \foreach \r/\y in {1/2.25,2/.75,3/-.75,4/-2.25} {
    \foreach \c/\x in {1/0,2/1.35,3/2.70,4/5.10,5/6.45,6/7.80} {
      \coordinate (v\r\c0) at ({\x-.11},{\y+.13});
      \coordinate (v\r\c1) at ({\x+.11},{\y-.13});
    }
  }

  \foreach \c in {1,2,3,4,5,6} {
    \foreach \r in {1,2,3,4} {
      \foreach \q in {1,2,3,4} \draw[edge] (v\r\c0) -- (v\q\c1);
    }
  }
  \foreach \r in {1,2,3,4} {
    \foreach \a/\b in {1/2,2/3,4/5,5/6} {
      \foreach \s in {0,1} {
        \foreach \t in {0,1} \draw[edge] (v\r\a\s) -- (v\r\b\t);
      }
    }
  }

  \foreach \r in {1,2,3,4} {
    \foreach \c in {1,2,3,4,5,6} {
      \node[side 0] at (v\r\c0) {};
      \node[side 1] at (v\r\c1) {};
    }
  }

  \foreach \y in {2.25,.75,-.75,-2.25} \node[fill=white, inner xsep=3pt] at (3.90,\y) {$\cdots$};

  \node[anchor=base] at (0,2.78) {$j=1$};
  \node[anchor=base] at (1.35,2.78) {$2$};
  \node[anchor=base] at (2.70,2.78) {$3$};
  \node[anchor=base] at (3.90,2.78) {$\cdots$};
  \node[anchor=base] at (5.10,2.78) {$2n-3$};
  \node[anchor=base] at (6.45,2.78) {$2n-2$};
  \node[anchor=base] at (7.80,2.78) {$2n-1$};

  \node[anchor=east] at (-.55,2.25) {$i=1$};
  \node[anchor=east] at (-.55,.75) {$2$};
  \node at (-.70,.12) {$\vdots$};
  \node[anchor=east] at (-.55,-.75) {$2m-1$};
  \node[anchor=east] at (-.55,-2.25) {$2m$};

  \node[side 0] at (2.65,-2.95) {};
  \node[anchor=west] at (2.8,-2.95) {$s=0$};
  \node[side 1] at (4.65,-2.95) {};
  \node[anchor=west] at (4.80,-2.95) {$s=1$};
\end{tikzpicture}
\par\medskip
\begin{tikzpicture}[
  edge/.style={draw=black, line width=.4pt, line cap=round},
  guide/.style={draw=black!24, densely dashed, line width=.35pt},
  vertex/.style={circle, draw=black, line width=.45pt, minimum size=3.7pt, inner sep=0pt},
  side 0/.style={vertex, fill=white},
  side 1/.style={vertex, fill=black}
]
  \foreach \r/\y in {0/1.15,1/-1.15} {
    \foreach \c/\x in {1/.95,2/1.90,3/2.85,4/3.80,5/5.35,6/6.30,7/7.25,8/8.20} {
      \coordinate (v\r\c0) at ({\x-.10},{\y+.15});
      \coordinate (v\r\c1) at ({\x+.10},{\y-.15});
    }
  }
  \coordinate (v000) at (0,1.30);
  \coordinate (v101) at (0,-1.30);
  \coordinate (v091) at (9.15,1.00);
  \coordinate (v190) at (9.15,-1.00);

  \draw[guide] (0,-1.65) -- (0,1.65);
  \draw[guide] (9.15,-1.65) -- (9.15,1.65);
  \node[anchor=south] at (0,1.72) {$C_{2j-1}$};
  \node[anchor=south] at (9.15,1.72) {$C_{2j'-1}$};
  \node[anchor=east] at (-.48,1.15) {$R_{e,0}$};
  \node[anchor=east] at (-.48,-1.15) {$R_{e,1}$};

  \foreach \c in {1,2,3,4,5,6,7,8} {
    \foreach \r in {0,1} {
      \foreach \q in {0,1} \draw[edge] (v\r\c0) -- (v\q\c1);
    }
  }
  \draw[edge] (v000) -- (v101);
  \draw[edge] (v190) -- (v091);

  \foreach \r in {0,1} {
    \foreach \a/\b in {1/2,2/3,3/4,5/6,6/7,7/8} {
      \foreach \s in {0,1} {
        \foreach \t in {0,1} \draw[edge] (v\r\a\s) -- (v\r\b\t);
      }
    }
  }
  \foreach \s in {0,1} {
    \draw[edge] (v000) -- (v01\s);
    \draw[edge] (v101) -- (v11\s);
    \draw[edge] (v08\s) -- (v091);
    \draw[edge] (v18\s) -- (v190);
  }

  \foreach \r in {0,1} {
    \foreach \c in {1,2,3,4,5,6,7,8} {
      \node[side 0] at (v\r\c0) {};
      \node[side 1] at (v\r\c1) {};
    }
  }
  \node[side 0] at (v000) {};
  \node[side 1] at (v101) {};
  \node[side 1] at (v091) {};
  \node[side 0] at (v190) {};

  \node[fill=white, inner xsep=3pt] at (4.575,1.15) {$\cdots$};
  \node[fill=white, inner xsep=3pt] at (4.575,-1.15) {$\cdots$};

  \node[side 0] at (3.45,-1.90) {};
  \node[anchor=west] at (3.58,-1.90) {$s=0$};
  \node[side 1] at (5.40,-1.90) {};
  \node[anchor=west] at (5.53,-1.90) {$s=1$};
\end{tikzpicture}
\caption{The intermediate graph $G$ (top) and the restriction of $G_F$ to the two rows reserved for $e=v_jv_{j'}$ (bottom).}
\label{mis:fig:construction}\label{app:mis:fig:intermediate-graph}\label{app:mis:fig:edge-gadget}
\end{figure}

\begin{lemma}\label{mis:exclusion}\label{app:mis:exclusion}
The intermediate graph $G$ and hence $G_F$ exclude the $5\times5$ grid as an induced minor.
\end{lemma}
\begin{proof}
    We first claim that every cycle of~$G$ contains two adjacent vertices in the same column but on different sides: vertices of the form $(i,j,0)$ and $(i',j,1)$.
    Deleting from $G$ all edges with endpoints in the same column but on different sides gives a~graph $G'$ that is the disjoint union of $2m$ graphs $Q_{2n-1}$ obtained from two $(2n-1)$-vertex paths $P_{2n-1}$ by making the $j$th vertex of the first path adjacent to the $(j-1)$st and $(j+1)$st vertices of the other path (when these vertices exist).
    Now consider a~\emph{leftmost} vertex $(\widehat{i},\widehat{j},s)$ (i.e., minimizing $\widehat{j}$) of a~cycle $C$ that does not contain a~pair $(i,j,0), (i',j,1)$.
    Its two neighbors on~$C$ must be $(\widehat{i},\widehat{j}+1,0)$ and $(\widehat{i},\widehat{j}+1,1)$.
    These two vertices form such a~pair.
    See the left side of~\cref{app:mis:fig:cycle-forced-pair} for an illustration.

    Let $H$ be any graph with three vertex-disjoint and nonadjacent cycles $\Gamma_1, \Gamma_2, \Gamma_3$ such that there is a~path between any two cycles avoiding the closed neighborhood of the third cycle.
    Note that the $5 \times 5$ grid is such a~graph~$H$; see the right side of~\cref{app:mis:fig:grid-three-cycles}.

    We claim that no such graph $H$ is an induced minor of~$G$.
    Suppose otherwise.
    The subgraph of~$G$ induced by the union $U_a$ of the branch sets representing $V(\Gamma_a)$ (for any $a \in [3]$) contains a~cycle.
    Contracting connected branch sets in a~forest cannot create a~cycle, whereas their contraction yields a~graph containing $\Gamma_a$.
    By the first paragraph, there are three column indices $j_1, j_2, j_3 \in [2n-1]$ such that $U_a$ contains a~pair of the form $(\cdot,j_a,0), (\cdot,j_a,1)$ for each $a \in [3]$.
    Without loss of generality, assume that $j_1 \leqslant j_2 \leqslant j_3$.
    Since the cycles $\Gamma_1, \Gamma_2, \Gamma_3$ are nonadjacent, we further have $j_1 < j_2 < j_3$.
    For any $i,i' \in [2m]$, the union of the closed neighborhoods of $(i,j_2,0)$ and $(i',j_2,1)$ in $G$ contains $C_{j_2}$.
    Removing $C_{j_2}$ disconnects $C_{j_1}$ from $C_{j_3}$.
    This contradicts the assumption that there is a~path between $\Gamma_1$ and $\Gamma_3$ in the non-neighborhood of~$\Gamma_2$.
  \end{proof}

\begin{figure}[t]
\centering
\begin{minipage}{.53\linewidth}
\centering
\begin{tikzpicture}[
  edge/.style={draw=black, line width=.32pt, line cap=round},
  cycle/.style={draw=black, line width=1pt, line cap=round, line join=round},
  vertex/.style={circle, draw=black, line width=.45pt,
    minimum size=3.7pt, inner sep=0pt},
  side 0/.style={vertex, fill=white},
  side 1/.style={vertex, fill=black},
  pair/.style={ellipse, draw=black, line width=.45pt,
    inner xsep=2.5pt, inner ysep=2pt}
]
  \foreach \c/\x in {0/0,1/1.25,2/2.50,3/3.75,4/5.00,5/6.25,6/7.50} {
    \coordinate (v\c0) at (\x,.24);
    \coordinate (v\c1) at (\x,-.24);
  }

  \foreach \a/\b in {0/1,1/2,2/3,3/4,4/5,5/6} {
    \foreach \s in {0,1} {
      \foreach \t in {0,1} \draw[edge] (v\a\s) -- (v\b\t);
    }
  }

  \draw[cycle] (v10) -- (v20) -- (v30) -- (v40) -- (v50) --
    (v41) -- (v31) -- (v21) -- cycle;

  \foreach \c in {0,1,2,3,4,5,6} {
    \node[side 0] (p\c0) at (v\c0) {};
    \node[side 1] (p\c1) at (v\c1) {};
  }

  \node[pair, fit=(p20)(p21)] {};
  \node[anchor=south] at (1.25,.58) {$\widehat{j}$};
  \node[anchor=south] at (2.50,.58) {$\widehat{j}+1$};
  \node[above] at (4.38,.42) {$C$};

  \node[side 0] at (2.70,-.90) {};
  \node[anchor=west] at (2.83,-.90) {$s=0$};
  \node[side 1] at (4.55,-.90) {};
  \node[anchor=west] at (4.68,-.90) {$s=1$};
\end{tikzpicture}
\end{minipage}\hfill
\begin{minipage}{.44\linewidth}
\centering
\begin{tikzpicture}[
  x=1cm,
  y=1cm,
  edge/.style={draw=black, line width=.35pt, line cap=round},
  cycle/.style={draw=black, line width=1.5pt, line cap=round, line join=round},
  connector/.style={draw=black, line width=1pt, line cap=round,
    preaction={draw=white, line width=1.7pt}},
  vertex/.style={circle, draw=black, fill=white, line width=.45pt,
    minimum size=3.5pt, inner sep=0pt}
]
  \foreach \x in {1,...,5} \draw[edge] (\x,1) -- (\x,5);
  \foreach \y in {1,...,5} \draw[edge] (1,\y) -- (5,\y);

  \draw[cycle] (1,4) -- (1,5) -- (2,5) -- (2,4) -- cycle;
  \draw[cycle] (4,4) -- (4,5) -- (5,5) -- (5,4) -- cycle;
  \draw[cycle] (1,1) -- (1,2) -- (2,2) -- (2,1) -- cycle;

  \draw[connector] (2,5) -- (3,5) -- (4,5);
  \draw[connector] (1,4) -- (1,3) -- (1,2);
  \draw[connector] (5,4) -- (5,3) -- (5,2) -- (5,1) -- (4,1) -- (3,1) -- (2,1);

  \foreach \x in {1,...,5} {
    \foreach \y in {1,...,5} \node[vertex] at (\x,\y) {};
  }

  \node at (1.5,4.5) {$\Gamma_1$};
  \node at (4.5,4.5) {$\Gamma_2$};
  \node at (1.5,1.5) {$\Gamma_3$};
  \node[anchor=south] at (3,5.08) {$P_{12}$};
  \node[anchor=east] at (.92,3) {$P_{13}$};
  \node[anchor=west] at (5.08,2.5) {$P_{23}$};
\end{tikzpicture}
\end{minipage}
\caption{The leftmost-vertex argument (left) and three pairwise nonadjacent cycles in the $5\times5$ grid with the required connecting paths (right).}
\label{mis:fig:exclusion}\label{app:mis:fig:cycle-forced-pair}\label{app:mis:fig:grid-three-cycles}
\end{figure}

We now relate the independence number of $G_F$ to the maximum cut of $F$.

\begin{lemma}\label{mis:identity}\label{app:mis:identity}
\[
 \alpha(G_F)=\maxcut(F)+
 2\sum_{\substack{v_jv_{j'}\in E(F)\\j<j'}}(j'-j).
\]
\end{lemma}
\begin{proof}
    We first show that \[\alpha(G_F) \geqslant \maxcut(F) + 2 \sum\limits_{v_j v_{j'} \in E(F),~j<j'}(j'-j).\]
    Let $(A_0,A_1)$ be a~bipartition of~$V(F)$ realizing the maximum cut.
    We identify $A_0$ (resp.~$A_1$) with the side $s=0$ (resp.~$s=1$).
    We build the independent set~$I$ as follows.

    For every $v_j \in A_0$ and every incident edge $e = v_jv_{j'} \in E(F)$, add to $I$ the unique vertex of $(R_{e,0} \cup R_{e,1}) \cap C_{2j-1} \cap S_0$.
    Note that this vertex is in $R_{e,0}$ if $j < j'$ and in $R_{e,1}$ if $j' < j$.
    Symmetrically, for every $v_j \in A_1$ and every incident edge $e = v_jv_{j'} \in E(F)$, add to $I$ the unique vertex of $(R_{e,0} \cup R_{e,1}) \cap C_{2j-1} \cap S_1$.

    At this stage, fix $e=v_jv_{j'} \in E(F)$ with $j < j'$.
    The set $I$ intersects $R_{e,0} \cup R_{e,1}$ in exactly two vertices: one in $C_{2j-1}$ and one in $C_{2j'-1}$.
    Suppose that the former (the vertex in~$C_{2j-1}$) is in $R_{e,s}$ and that the latter is in $R_{e,s'}$, with $s, s' \in \{0,1\}$.
    We add the following vertices of $R_{e,0} \cup R_{e,1}$ to the independent set~$I$:
    \begin{itemize}
    \item for every $v_{j''} \in A_0$ with $j'' \in [j+1,j'-1]$, we add the vertex $(e^s,2j''-1,0)$;
    \item for every $v_{j''} \in A_1$ with $j'' \in [j+1,j'-1]$, we add the vertex $(e^s,2j''-1,1)$;
    \item for every $j'' \in [j,j'-2]$, we add $(e^{1-s},2j'',0)$ to $I$;
    \item in addition, if $s=s'$, we add $(e^{1-s},2j'-2,0)$ to~$I$.
    \end{itemize}
    
    (Note that for the last two items we could also have consistently picked the corresponding vertices in~$S_1$.
    In fact, we could have kept only a~single vertex at the intersection of a~row with a~column of even index, but we decided against that to simplify the description of~$G$.)
    This finishes the construction of~$I$.
    It is an independent set by the following two properties:
    \begin{itemize}
    \item in every column, we only picked vertices from one side, and
    \item we never picked two vertices in consecutive columns of the same row.
    \end{itemize}
    
    We now determine the size of~$I$.
    For every $e=v_jv_{j'} \in E(F)$ with $j<j'$, $I \cap (R_{e,0} \cup R_{e,1})$ contains exactly one vertex in every column intersected by $R_{e,0} \cup R_{e,1}$ except possibly $C_{2j'-2}$, in which it contains at most one vertex.
    Thus, \[|I \cap (R_{e,0} \cup R_{e,1})| \in \{2(j'-j),2(j'-j)+1\}.\]
    Furthermore, $e$ is in the cut $(A_0,A_1)$ if and only if $s=s'$.
    Thus, cut edges contribute $2(j'-j)+1$ vertices and all other edges contribute $2(j'-j)$ vertices.
    See \cref{app:mis:fig:alpha-lower-bound} for an illustration.
    Therefore, \[|I| = \maxcut(F) + 2 \sum\limits_{v_j v_{j'} \in E(F),~j<j'}(j'-j),~~\text{as desired}.\]

\begin{figure}[!ht]
\centering
\begin{tikzpicture}[
  edge/.style={draw=black, line width=.3pt, line cap=round},
  vertex/.style={circle, draw=black, line width=.4pt, minimum size=3.2pt, inner sep=0pt},
  side 0/.style={vertex, fill=white},
  side 1/.style={vertex, fill=black},
  chosen/.style={circle, draw=black, line width=.7pt, minimum size=5.5pt, inner sep=0pt},
  panel label/.style={font=\small},
  count/.style={font=\small, text=blue}
]
  \foreach \p/\X/\a/\b in {c/0/0/1,n/7.10/0/0} {
    \begin{scope}[xshift={\X cm}]
    \foreach \r/\y in {0/.38,1/-.38} {
      \foreach \c/\x in {1/.95,2/1.90,3/2.85,4/3.80,5/4.75} {
        \coordinate (v\p\r\c0) at ({\x-.075},{\y+.105});
        \coordinate (v\p\r\c1) at ({\x+.075},{\y-.105});
      }
    }
    \coordinate (l\p0) at (0,.49);
    \coordinate (l\p1) at (0,-.49);
    \coordinate (r\p0) at (5.70,.26);
    \coordinate (r\p1) at (5.70,-.26);

    \foreach \c in {1,2,3,4,5} {
      \foreach \r in {0,1} {
        \foreach \q in {0,1} \draw[edge] (v\p\r\c0) -- (v\p\q\c1);
      }
    }
    \draw[edge] (l\p0) -- (l\p1);
    \draw[edge] (r\p0) -- (r\p1);
    \foreach \r in {0,1} {
      \foreach \a/\b in {1/2,2/3,3/4,4/5} {
        \foreach \s in {0,1} {
          \foreach \t in {0,1} \draw[edge] (v\p\r\a\s) -- (v\p\r\b\t);
        }
      }
      \foreach \s in {0,1} {
        \draw[edge] (l\p\r) -- (v\p\r1\s);
        \draw[edge] (v\p\r5\s) -- (r\p\r);
      }
    }

    \foreach \r in {0,1} {
      \foreach \c in {1,2,3,4,5} {
        \node[side 0] (p\p\r\c0) at (v\p\r\c0) {};
        \node[side 1] (p\p\r\c1) at (v\p\r\c1) {};
      }
    }
    \node[side 0] (pl\p0) at (l\p0) {};
    \node[side 1] (pl\p1) at (l\p1) {};
    \node[side 1] (pr\p0) at (r\p0) {};
    \node[side 0] (pr\p1) at (r\p1) {};

    \node[panel label] at (2.85,1.15) {$v_j\in A_{\a},\ v_{j'}\in A_{\b}$};
    \node[panel label, anchor=south] at (0,.72) {$C_{2j-1}$};
    \node[panel label, anchor=south] at (5.70,.72) {$C_{2j'-1}$};
    \end{scope}
  }

  \foreach \u in {plc0,pc110,pc021,pc130,pc040,pc150,prc0} \node[chosen, fit=(\u)] {};
  \foreach \u in {pln0,pn021,pn040,prn1,pn110,pn130} \node[chosen, fit=(\u)] {};

  \node[count] at (2.85,-1.02) {$|I\cap(R_{e,0}\cup R_{e,1})|=2(j'-j)+1$};
  \node[count] at (9.95,-1.02) {$|I\cap(R_{e,0}\cup R_{e,1})|=2(j'-j)$};

\end{tikzpicture}
\caption{Example of $I\cap(R_{e,0}\cup R_{e,1})$ (circled vertices) when $v_j$ and $v_{j'}$ are on opposite sides of the bipartition $(A_0,A_1)$ (left) and when they are on the same side (right).}
\label{app:mis:fig:alpha-lower-bound}
\end{figure}

    We now show that \[\alpha(G_F) \leqslant \maxcut(F) + 2 \sum\limits_{v_j v_{j'} \in E(F),~j<j'}(j'-j).\]
    Fix an independent set~$I$ of~$G_F$.
    In every column $C_{2j-1}$ of $v_j$, $I$ can contain vertices from only one side (because of the biclique between $S_0 \cap C_{2j-1}$ and $S_1 \cap C_{2j-1}$).
    The set $I$ thus defines a~partition $(A_0,A_1)$ of $V(F)$ in which $v_j \in V(F)$ is in~$A_0$ if $I \cap C_{2j-1} \subseteq S_0$ and $v_j \in A_1$ if instead $I \cap C_{2j-1} \subseteq S_1$ (vertices whose columns do not intersect $I$ are assigned arbitrarily to either part).
    For every $e=v_jv_{j'} \in E(F)$ with $j<j'$, we partition $R_{e,0} \cup R_{e,1}$ into $D_j, D_{j+1}, \ldots, D_{j'-1}, D_{j'}$, where
    \begin{itemize}
    \item $D_j := C_{2j-1} \cap (R_{e,0} \cup R_{e,1})$; and
    \item for every $j'' \in [j+1,j']$, $D_{j''} := (C_{2j''-2} \cup C_{2j''-1}) \cap (R_{e,0} \cup R_{e,1})$.
    \end{itemize}

    Since $D_j$ is a clique and each of the other parts is the union of two cliques, $I$ can contain at most one vertex in $D_j$ and at most two vertices in each $D_{j''}$ with $j'' \in [j+1,j']$.
    This implies that $|I \cap (R_{e,0} \cup R_{e,1})| \leqslant 1+2(j'-j)$; see~\cref{app:mis:fig:alpha-upper-bound}.
    \begin{figure}[!ht]
\centering
\begin{tikzpicture}[
  edge/.style={draw=black, line width=.32pt, line cap=round},
  vertex/.style={circle, draw=black, line width=.45pt, minimum size=3.7pt, inner sep=0pt},
  side 0/.style={vertex, fill=white},
  side 1/.style={vertex, fill=black},
  block/.style={draw=black!45, densely dashed, rounded corners=1.5pt, inner sep=3pt},
  chosen/.style={circle, draw=black, line width=.8pt, minimum size=6.4pt, inner sep=0pt}
]
  \foreach \r/\y in {0/.43,1/-.43} {
    \foreach \c/\x in {1/1.10,2/2.20,3/3.45,4/4.55,5/6.25} {
      \coordinate (v\r\c0) at ({\x-.09},{\y+.13});
      \coordinate (v\r\c1) at ({\x+.09},{\y-.13});
    }
  }
  \coordinate (l0) at (0,.56);
  \coordinate (l1) at (0,-.56);
  \coordinate (r0) at (7.35,.30);
  \coordinate (r1) at (7.35,-.30);

  \foreach \c in {1,2,3,4,5} {
    \foreach \r in {0,1} {
      \foreach \q in {0,1} \draw[edge] (v\r\c0) -- (v\q\c1);
    }
  }
  \draw[edge] (l0) -- (l1);
  \draw[edge] (r0) -- (r1);
  \foreach \r in {0,1} {
    \foreach \a/\b in {1/2,2/3,3/4} {
      \foreach \s in {0,1} {
        \foreach \t in {0,1} \draw[edge] (v\r\a\s) -- (v\r\b\t);
      }
    }
    \foreach \s in {0,1} {
      \draw[edge] (l\r) -- (v\r1\s);
      \draw[edge] (v\r5\s) -- (r\r);
    }
  }

  \foreach \r in {0,1} {
    \foreach \c in {1,2,3,4,5} {
      \node[side 0] (p\r\c0) at (v\r\c0) {};
      \node[side 1] (p\r\c1) at (v\r\c1) {};
    }
  }
  \node[side 0] (pl0) at (l0) {};
  \node[side 1] (pl1) at (l1) {};
  \node[side 1] (pr0) at (r0) {};
  \node[side 0] (pr1) at (r1) {};

  \node[block, fit=(pl0)(pl1)] (Dj) {};
  \node[block, fit=(p010)(p011)(p110)(p111)(p020)(p021)(p120)(p121)] (Djone) {};
  \node[block, fit=(p030)(p031)(p130)(p131)(p040)(p041)(p140)(p141)] (Djtwo) {};
  \node[block, fit=(p050)(p051)(p150)(p151)(pr0)(pr1)] (Djp) {};

  \node[anchor=south] at (Dj.north) {$D_j$};
  \node[anchor=south] at (Djone.north) {$D_{j+1}$};
  \node[anchor=south] at (Djtwo.north) {$D_{j+2}$};
  \node[anchor=south] at (Djp.north) {$D_{j'}$};
  \node[anchor=north] at (Dj.south) {$\leqslant 1$};
  \node[anchor=north] at (Djone.south) {$\leqslant 2$};
  \node[anchor=north] at (Djtwo.south) {$\leqslant 2$};
  \node[anchor=north] at (Djp.south) {$\leqslant 2$};

  \foreach \u in {pl0,p110,p021,p130,p040,p150,pr0} \node[chosen, fit=(\u)] {};
  \node[fill=white, inner xsep=3pt] at (5.40,.43) {$\cdots$};
  \node[fill=white, inner xsep=3pt] at (5.40,-.43) {$\cdots$};
  \node[anchor=north] at (3.68,-1.22)
    {\textcolor{blue}{$|I\cap(R_{e,0}\cup R_{e,1})|\leqslant 1+2(j'-j)$}};

\end{tikzpicture}
\caption{The partition $D_j, \ldots, D_{j'}$ and an upper bound on $|I\cap(R_{e,0}\cup R_{e,1})|$.}
\label{app:mis:fig:alpha-upper-bound}
    \end{figure}
    
    Furthermore, if $v_j$ and $v_{j'}$ are on the same side of $(A_0,A_1)$, we claim that \[|I \cap (R_{e,0} \cup R_{e,1})| \leqslant 2(j'-j).\]
    In each of the two rows, at~least one endpoint cannot belong to $I$ because the endpoints are on opposite sides.
    After deleting that forbidden endpoint, the row occupies $2(j'-j)$ consecutive columns.
    Pairing these columns consecutively partitions the remaining row into $j'-j$ cliques.
    Hence, $I$ contains at~most $j'-j$ vertices in each row.

    Summing over all edges and using that $(A_0,A_1)$ has at most $\maxcut(F)$ cut edges, we conclude that $|I| \leqslant \maxcut(F) + 2 \sum\limits_{v_j v_{j'} \in E(F),~j<j'}(j'-j)$.
  \end{proof}

The additive term in \cref{mis:identity} is computable in polynomial time.
Hence, solving \smis{} on $G_F$ is as hard as solving \textsc{Max Cut} on $F$.
Together with \cref{mis:exclusion}, this proves \cref{thm:mis}.

\paragraph{The WETH lower bound.}
When graph $F$ has maximum degree at~most~$3$, the construction has quadratic size.
This yields the following lower bound.

\misLowerBound*
We first establish the corresponding lower bound for \textsc{Max Cut} on subcubic graphs.

\begin{lemma}\label{app:mis:subcubic-cut}
  Unless WETH fails, for every fixed $\varepsilon > 0$, \textnormal{\textsc{Max Cut}} on $n$-vertex simple graphs of maximum degree at~most~$3$ has no $2^{\Oh(n^{1-\varepsilon})}$-time algorithm.
\end{lemma}
\begin{proof}
  We give a linear-size reduction from \textsc{3-SAT} with $N$ variables and $M$ clauses to an instance $(F,K)$ of \textsc{Max Cut} with $\Oh(N+M)$ vertices.
  Clauses with one or two literals can be extended to three literals by repeating a literal.
  Therefore, we may write every clause as $a \lor b \lor c$, allowing repeated literals.

  A not-all-equal clause is satisfied when the truth values of its three literals are not all equal.
  Introduce one common variable $f$.
  For each clause $a \lor b \lor c$, introduce a fresh variable $z$ and the two not-all-equal clauses
  \[
    \nae(a,b,z),\qquad
    \nae(\neg z,c,f).
  \]
  Denote the resulting formula by $\Psi$ and its number of clauses by $q=2M$.
  Complementing all variable values preserves every not-all-equal clause, so a satisfying assignment of $\Psi$ can be assumed to have $f=0$.
  With this value of $f$ and $c=0$, the second clause forces $z=0$ and the first clause is equivalent to $a\lor b$.
  If $c=1$, the second clause is satisfied for either value of $z$ and a value satisfying the first clause always exists.
  Consequently, the original formula is satisfiable if and only if $\Psi$ is satisfiable.

  For each variable $x$ occurring in $\Psi$, let $p_x$ and $n_x$ be its positive and negative occurrence counts and put $t_x=\max\{p_x,n_x,2\}$.
  Make a cycle $C_x$ of length $2t_x$.
  Its vertices in even positions and odd positions will represent the positive and negative occurrences of $x$, respectively.
  Assign every positive occurrence of $x$ to a distinct even vertex of $C_x$ and every negative occurrence to a distinct odd vertex.
  For each clause of $\Psi$, make a triangle with one vertex for each of its three literal occurrences and join each such vertex to the cycle vertex assigned to that occurrence.
  This defines~$F$.
  All assigned cycle vertices are distinct within each cycle, including when the same literal occurs more than once in a~clause.
  Every cycle has length at~least~$4$, so $F$ is simple.
  Cycle vertices have degree~$2$ or $3$ and triangle vertices have degree~$3$.

  Write $S=\sum_x2t_x$ and set
  \[
    K=S+3q+2q=S+5q.
  \]
  Any cut of~$F$ contains at most $S$ cycle edges, at most $3q$ edges between cycles and triangles and at most two edges of each clause triangle.
  Hence, every cut has size at most $K$.
  If $\Psi$ is satisfiable, color the even vertices of $C_x$ with the truth value of $x$ and its odd vertices with the opposite value.
  Color each triangle vertex with the color opposite to the truth value of its literal.
  All cycle edges and all edges joining cycles to triangles are cut.
  Each clause triangle has both colors and contributes two cut edges.
  This gives a~cut of size~$K$.
  Conversely, a~cut of size at least $K$ must attain equality in every one of the preceding bounds.
  Therefore, each cycle is properly two-colored, giving a truth value to its variable from its even vertices.
  Every triangle vertex has the opposite color from the truth value of its literal.
  The two cut edges in each triangle imply that its three literal values are not all equal.
  These variable values satisfy $\Psi$.
  Thus, $F$ has a cut of size at least $K$ if and only if the original formula is satisfiable.

  Finally, set $r_x:=p_x+n_x$.
  Since $r_x\geq1$, we have $t_x\leq2r_x$; since $\sum_xr_x=3q$, it follows that $S\leq12q$.
  Therefore, $|V(F)|=S+3q\leq15q=30M$ and the construction takes $\Oh(N+M)$ time.

  To derive the claimed lower bound from WETH, we use the quantitative form of the Sparsification Lemma~\cite[Section~2]{sparsification}.
  We may first remove repeated clauses, leaving $M=\Oh(N^3)$.
  Apply its proof with $3$-clauses and thresholds $h_1=\tau$ and $h_2=4\tau^2$, where $\tau\ge2$ is an integer.

  Lemmas~2 and~7 of that proof, applied to the universe of $2N$ literals, yield at~most $2^{\Oh\left(\frac{N\log\tau}{\tau}\right)}$ formulas, each with at~most $N$ variables and $\Oh(\tau^2N)$ clauses.
  These estimates hold uniformly in~$\tau$.
  Indeed, each branching weak sunflower can be found by enumerating the possible cores of size $1$ or~$2$ and checking the number and common intersection of the clauses of each size containing the core.
  This takes polynomial time per node, with an exponent independent of~$\tau$, and the search tree is binary.
  Thus the time spent in sparsifying (i.e., running the algorithmic proof of the Sparsification Lemma) is $2^{\Oh\left(\frac{N\log\tau}{\tau}\right)}(N+M+\tau)^{\Oh(1)}$.

  Suppose that \textsc{Max Cut} on subcubic graphs can be solved in time $2^{\Oh(n^{1-\varepsilon})}$ for some fixed $0<\varepsilon<1$.
  Set $\tau:=\max\{2,\lceil N^{\varepsilon/4}\rceil\}$.
  Applying the reduction above to each sparsified formula gives graphs with $\Oh(N^{1+\varepsilon/2})$ vertices.
  Solving all these instances decides the original formula in time
  \[
    2^{\Oh(N^{1-\varepsilon/4}\log N+N^{(1+\varepsilon/2)(1-\varepsilon)})}
    =2^{\Oh(N^{1-\varepsilon/8})},
  \]
  where we use $(1+\varepsilon/2)(1-\varepsilon)\le1-\varepsilon/2$ and $\log N=\Oh(N^{\varepsilon/8})$.
  This contradicts WETH.
\end{proof}

\begin{proof}[Proof of \cref{mis:eth}]
  Let $F$ be a simple subcubic graph with $n$ vertices and $m$ edges.
  For each $e=v_jv_{j'}$ with $j<j'$, its two rows contribute exactly $8(j'-j)$ vertices to $G_F$.
  Therefore,
  \[
    |V(G_F)|=8\sum_{v_jv_{j'}\in E(F),\,j<j'}(j'-j)
    \le8m(n-1)\le12n(n-1)\le12n^2.
  \]
  By \cref{app:mis:identity}, we can convert any \textsc{Max Cut} threshold for $F$ into a \textsc{MIS} threshold for $G_F$ in polynomial time.
  \cref{app:mis:exclusion} guarantees that $G_F$ excludes the $5\times5$ grid as an induced minor.
  Therefore, a $2^{\Oh(N^{1/2-\varepsilon})}$-time algorithm for \textsc{MIS} on these $N$-vertex graphs, for some fixed \mbox{$0<\varepsilon<1/2$}, would solve \textsc{Max Cut} on $F$ in time $2^{\Oh(n^{1-2\varepsilon})}$, contradicting~\cref{app:mis:subcubic-cut} under WETH.
\end{proof}

\section{Hardness of \textsc{3-Coloring} and \textsc{Feedback Vertex Set}}
\label{sec:further-hardness}

\cref{fig:red-components} shows the forms of red components in the contraction sequences for the three reductions below.
Each illustrated form has max-leaf at most $3$, so these sequences witness reduced component max-leaf at most $3$.
The reductions use pairwise anti-complete rows of alternating $A$- and $B$-blocks and independent vertices $h_j$.
Each $h_j$ is complete to the $j$th $B$-block of every row and has no other neighbors.
While we contract parts within each block and then corresponding parts of different rows, every $B_j$-part with $j\ge1$ remains complete to $h_j$ and every $A$-part remains anti-complete to all $h_j$.
Thus, no red edge is incident with an $h_j$.

Once each block forms one part, we contract corresponding parts from left to right one row at a~time.
After each contraction, the red component containing the newly obtained part consists of at~most three paths meeting only at that part.
Thus, every red component is a path or a subdivision of $K_{1,3}$.

\begin{figure}[!ht]
\centering
\begin{tikzpicture}[
  x=1cm,y=.85cm,
  red edge/.style={draw=red!75!black,line width=.9pt},
  part/.style={circle,draw=black,fill=white,inner sep=0pt,minimum size=4pt}
]
  \foreach \i/\x in {0/0,1/.7,2/1.4,3/2.1,4/2.8} {
    \node[part] (p\i) at (\x,.25) {};
  }
  \draw[red edge] (p0)--(p1)--(p2)--(p3)--(p4);
  \node[anchor=north] at (1.4,-.75) {path};

  \node[part] (c) at (6.7,.25) {};
  \node[part] (l1) at (5.95,.25) {};
  \node[part] (l2) at (5.2,.25) {};
  \node[part] (u1) at (7.3,.6) {};
  \node[part] (u2) at (7.9,.95) {};
  \node[part] (d1) at (7.3,-.1) {};
  \node[part] (d2) at (7.9,-.45) {};
  \draw[red edge] (l2)--(l1)--(c)--(u1)--(u2);
  \draw[red edge] (c)--(d1)--(d2);
  \node[anchor=north] at (6.55,-.75) {subdivision of $K_{1,3}$};

  \node[part] (a) at (11.5,0) {};
  \node[part] (b) at (12,.95) {};
  \node[part] (d) at (12.5,0) {};
  \node[part] (a1) at (10.85,0) {};
  \node[part] (a2) at (10.2,0) {};
  \node[part] (d1x) at (13.15,0) {};
  \node[part] (d2x) at (13.8,0) {};
  \draw[red edge] (a)--(b)--(d)--(a);
  \draw[red edge] (a2)--(a1)--(a);
  \draw[red edge] (d)--(d1x)--(d2x);
  \node[anchor=north] at (12,-.75) {$K_3$ with two pendant paths};
\end{tikzpicture}
\caption{Forms of red components in the contraction sequences.}
\label{fig:red-components}
\end{figure}

\subsection{\textsc{3-Coloring}}
\label{sec:coloring}\label{app:col}\label{app:coloring}

A hypergraph $\mathcal{H} = (V, \mathcal{F})$ is $3$-uniform if every hyperedge in $\mathcal{F}$ has size exactly $3$.
A~$2$-coloring of $\mathcal{H}$ is a mapping $\varphi: V(\mathcal{H}) \to [2]$ such that no hyperedge is monochromatic.
The \textnormal{\textsc{$3$-Uniform Hypergraph $2$-Colorability}} problem asks whether a given $3$-uniform hypergraph admits a~\mbox{$2$-coloring}.
The problem is known to be NP-complete~\cite{Lovasz73}.

\coloringHardness*

\begin{proof}
    We reduce from \textnormal{\textsc{$3$-Uniform Hypergraph $2$-Colorability}}.
    Let $\mathcal{H} = (V, \mathcal{F})$ be an instance where $V = \{v_1, \ldots, v_n\}$ and $\mathcal{F} = \{F_1, \ldots, F_m\}$.
    We may assume that $m\ge1$.

\proofparagraph{The construction.}
    We construct a graph $G_\mathcal{H}$ such that $G_\mathcal{H}$ is $3$-colorable if and only if $\mathcal{H}$ is $2$-colorable.
    Given three pairwise nonadjacent vertices $r_1,r_2,r_3$, let $\nae(r_1,r_2,r_3)$ be the graph obtained from a $K_3$ on vertices $k_1,k_2,k_3$ by adding the edges $r_ik_i$ for all $i\in[3]$.

    \begin{claim}\label{app:col:claim:tww2NAE}
        A coloring of $r_1,r_2,r_3$ extends to a $3$-coloring of $\nae(r_1,r_2,r_3)$ if and only if $r_1,r_2,r_3$ are not all assigned the same color.
    \end{claim}
    \begin{proof}[Proof of the claim]
        Suppose first that $r_1,r_2,r_3$ are assigned the same color.
        Since $r_i$ is adjacent to $k_i$ for every $i\in[3]$, the vertices $k_1,k_2,k_3$ can use only the other two colors.
        Since $k_1,k_2,k_3$ induce a $K_3$, the coloring does not extend to a $3$-coloring of $\nae(r_1,r_2,r_3)$.

        Conversely, suppose that $r_1,r_2,r_3$ are not all assigned the same color.
        If they use two colors, we may assume that $r_1$ and $r_2$ have color $1$ and $r_3$ has color $2$.
        We assign colors $2,0,1$ to $k_1,k_2,k_3$, respectively.
        If $r_1,r_2,r_3$ have three distinct colors, we may assume that they have colors $1,2,0$, respectively.
        Again, we assign colors $2,0,1$ to $k_1,k_2,k_3$, respectively.
        Therefore, the coloring extends in both cases to a $3$-coloring of $\nae(r_1,r_2,r_3)$.
    \end{proof}

    We first introduce independent vertices $h_1,\ldots,h_n$, one for each vertex of $\mathcal{H}$.
    For each hyperedge $F_i \in \mathcal{F}$, we introduce a row $R_i$ consisting of the blocks
    \[A_{i,1} - B_{i,1} - A_{i,2} - B_{i,2} - \cdots - B_{i, n} - A_{i, n+1}.\]
    
    Within each row, edges between distinct blocks occur only between consecutive blocks.
    Distinct rows are pairwise anti-complete.
    Fix $F_i=\{v_a,v_b,v_c\}$ with $a<b<c$.
    The vertices $h_a,h_b,h_c$ represent the colors of $v_a,v_b,v_c$.
    In every $3$-coloring, the vertices $\alpha_{i,k}$ introduced below have a~common color within $R_i$ that no $h_k$ uses.
    Two propagations carry the colors determined by $h_a$ and $h_c$ to $A_{i,b}$, where the coloring extends to the NAE gadget only if $h_a,h_b,h_c$ are not monochromatic.
    We construct $R_i$ as follows:
    
    \proofparagraph{(1) The block $\bm{A_{i,k}}$.}
    
    For each $k\in[n+1]$, we introduce a vertex $\alpha_{i,k}$ in $A_{i,k}$.
    For each $k \in [n+1]$, we also introduce vertices $\beta_{i,k}$ and $\gamma_{i,k}$ in $A_{i,k}$ (omitting $\gamma_{i,1}$ and $\beta_{i,n+1}$) and add the edges $\alpha_{i,k}\beta_{i,k}$ and $\alpha_{i,k}\gamma_{i,k}$ whenever the corresponding vertices exist.
    For each $k\in \{a,b\}$, we introduce a vertex $\lambda_{i,k}$ in $A_{i,k}$ and add the edges $\lambda_{i,k}\alpha_{i,k}$ and $\lambda_{i,k}\beta_{i,k}$.
    For $k = c$, we introduce a vertex $\lambda_{i, c}$ in $A_{i, c+1}$ and add the edges $\lambda_{i, c}\alpha_{i, c+1}$ and $\lambda_{i, c}\gamma_{i, c+1}$.
    
    For each $k\in[a+1,b]$, we introduce two vertices $p^1_{i,k}$ and $q^1_{i,k}$ in $A_{i,k}$ such that $\alpha_{i,k}$, $p^1_{i,k}$ and $q^1_{i,k}$ are pairwise adjacent.
    Similarly, for each $k\in[b, c]$, we introduce two vertices $p^2_{i,k}$ and $q^2_{i,k}$ in $A_{i,k}$ such that $\alpha_{i,k}$, $p^2_{i,k}$ and $q^2_{i,k}$ are pairwise adjacent.
    Finally, in $A_{i,b}$, we introduce three vertices $s^1_i,s^2_i,s^3_i$ inducing a $K_3$ and add the edges $\lambda_{i,b}s^1_i$, $p^1_{i,b}s^2_i$ and $p^2_{i,b}s^3_i$.
    These six vertices induce $\nae(\lambda_{i,b},p^1_{i,b},p^2_{i,b})$, with $s^1_i,s^2_i,s^3_i$ corresponding to $k_1,k_2,k_3$, respectively.

    \proofparagraph{(2) The block $\bm{B_{i,k}}$.}

    For each $k\in[n]$, we introduce vertices $w_{i,k}$, $z_{i,k}$ and $\mu_{i,k}$ in $B_{i,k}$ and add the edge $w_{i,k}z_{i,k}$.
    For each $k\in[a,b-1]$, we introduce vertices $x^1_{i,k}$ and $y^1_{i,k}$ in $B_{i,k}$ and add the edge $x^1_{i,k}y^1_{i,k}$.
    Similarly, for each $k\in[b,c]$, we introduce vertices $x^2_{i,k}$ and $y^2_{i,k}$ in $B_{i,k}$ and add the edge $x^2_{i,k}y^2_{i,k}$.

    \proofparagraph{(3) Edges between consecutive blocks.}

    For every $k\in[n]$, we add the edges $\beta_{i,k}w_{i,k}$, $\beta_{i,k}z_{i,k}$ and $\alpha_{i,k}\mu_{i,k}$ between $A_{i,k}$ and $B_{i,k}$.
    We add the edges $w_{i,k}\gamma_{i,k+1}$, $z_{i,k}\gamma_{i,k+1}$ and $\mu_{i,k}\alpha_{i,k+1}$ between $B_{i,k}$ and $A_{i,k+1}$.

    For every $k \in [a, b-1]$, we add the edges $x^1_{i,k}p^1_{i,k+1}$ and $y^1_{i,k}q^1_{i,k+1}$ between $B_{i,k}$ and $A_{i,k+1}$.
    For every $k \in [a+1,b-1]$, we also add the edges $p^1_{i,k}x^1_{i,k}$ and $q^1_{i,k}y^1_{i,k}$ between $A_{i,k}$ and $B_{i,k}$, while for $k=a$ we instead add the edge $\lambda_{i,a}x^1_{i,a}$.

    For every $k \in [b, c]$, we add the edges $p^2_{i,k}x^2_{i,k}$ and $q^2_{i,k}y^2_{i,k}$ between $A_{i,k}$ and $B_{i,k}$.
    For every $k \in [b, c-1]$, we also add the edges $x^2_{i,k}p^2_{i,k+1}$ and $y^2_{i,k}q^2_{i,k+1}$ between $B_{i,k}$ and $A_{i,k+1}$, while for $k = c$ we instead add the edge $x^2_{i,c}\lambda_{i,c}$.

    \medskip

    This completes the construction of $R_i$.
    For each $j\in[n]$, we add all edges between $h_j$ and $B_{i,j}$ and no edges between $h_j$ and $V(R_i)\setminus B_{i,j}$.
    All unspecified pairs are nonadjacent.
    This completes the construction of $G_\mathcal{H}$.
    Each block has constant size, so the graph has $\Oh(mn+n)$ vertices.

    \cref{app:col:fig:construction} illustrates the construction.
    \cref{app:col:fig:coloring} shows the same graph with a proper $3$-coloring.

    \begin{figure}[htbp]
        \centering
        \ThreeColoringConstruction{0}
        \caption{The row $R_i$ for $F_i=\{v_2,v_4,v_5\}$, together with $h_1,\ldots,h_6$.
        A stem from $h_j$ to $B_{i,j}$ represents adjacency to every vertex of that block.
        The row index $i$ is omitted from vertex labels.}
        \label{app:col:fig:construction}
    \end{figure}

    \begin{figure}[htbp]
        \centering
        \ThreeColoringConstruction{1}
        \caption{A proper $3$-coloring of the graph in~\cref{app:col:fig:construction}.
        Gray, blue and orange represent colors $0,1,2$, respectively.
        In particular, $h_2,h_4,h_5$ have colors $1,1,2$.}
        \label{app:col:fig:coloring}
    \end{figure}

\proofparagraph{Correctness.}
    We prove the correctness of the reduction.
    Suppose first that $G_\mathcal{H}$ is $3$-colorable and let $\eta: V(G_\mathcal{H}) \to \{0, 1, 2\}$ be a $3$-coloring of $G_\mathcal{H}$.
    Fix a row $R_i$.
    Let $F_i=\{v_a,v_b,v_c\}$ with $a <  b  < c$.
    We first determine the colors of the vertices $\beta_{i, k}$ and $\gamma_{i, k+1}$.

    \begin{claim}\label{app:col:claim:diamondcolor}
        For every $k \in [n]$, we have $\eta(\beta_{i, k}) = \eta(h_k) = \eta(\gamma_{i, k+1})$.
    \end{claim}
    \begin{proof}[Proof of the claim]
        Since $w_{i,k}z_{i,k} \in E(G_{\mathcal{H}})$, the vertices $w_{i,k}$ and $z_{i,k}$ have distinct colors.
        Both $\beta_{i,k}$ and $h_k$ are adjacent to $w_{i,k}$ and $z_{i,k}$.
        Hence, $\eta(\beta_{i,k})=\eta(h_k)$.
        
        Similarly, since $\gamma_{i,k+1}$ and $h_k$ are both adjacent to $w_{i,k}$ and $z_{i,k}$, we obtain $\eta(\gamma_{i,k+1})=\eta(h_k)$.
        Therefore, $\eta(\beta_{i,k})=\eta(h_k)=\eta(\gamma_{i,k+1})$.
    \end{proof}

    We next show that all vertices $\alpha_{i,k}$ in the same row have the same color.

    \begin{claim}\label{app:col:claim:allsamecolor}
        $\eta(\alpha_{i, 1}) = \eta(\alpha_{i, 2}) = \cdots = \eta(\alpha_{i, n+1})$.
    \end{claim}
    \begin{proof}[Proof of the claim]
        Fix $i\in[m]$ and $k\in[n]$.
        By \cref{app:col:claim:diamondcolor}, we have $\eta(\beta_{i,k})=\eta(h_k)=\eta(\gamma_{i,k+1})$.
        Since $\alpha_{i,k}\beta_{i,k}, \alpha_{i,k+1}\gamma_{i,k+1} \in E(G_\mathcal{H})$, we have $\eta(\alpha_{i, k}) \neq \eta(h_k)$ and $\eta(\alpha_{i, k+1}) \neq \eta(h_k)$.

        Both $\alpha_{i,k}$ and $\alpha_{i,k+1}$ also avoid the color of $\mu_{i,k}$.
        Since $h_k\mu_{i,k}$ is an edge, the colors of $h_k$ and $\mu_{i,k}$ are distinct.
        Thus, $\alpha_{i,k}$ and $\alpha_{i,k+1}$ have the same remaining color.
        Applying this argument for every $k\in[n]$ gives $\eta(\alpha_{i,1})=\eta(\alpha_{i,2})=\cdots=\eta(\alpha_{i,n+1})$.
    \end{proof}

    By \cref{app:col:claim:diamondcolor,app:col:claim:allsamecolor}, the vertices $\alpha_{i,1},\ldots,\alpha_{i,n+1}$ have one common color that differs from the color of every $h_k$.
    Since $m\ge1$, applying this observation to any one row shows that $h_1,\ldots,h_n$ use at most two colors in the original coloring.
    For the argument within the fixed row $R_i$, rename the colors so that $\eta(\alpha_{i,k})=0$ for every $k\in[n+1]$.
    Then, $\eta(h_k)\in\{1,2\}$ for every $k\in[n]$.
    This choice of color names is only for the argument in $R_i$ and does not change whether any three vertices $h_k$ are monochromatic.

    We now describe how the color of $h_a$ propagates from $A_{i,a}$ to $A_{i,b}$.
    Every pair $p^r_{i,k},q^r_{i,k}$ introduced in the construction forms a triangle with $\alpha_{i,k}$, so its vertices use colors $1,2$.
    \begin{claim}\label{app:col:claim:propagetefromatob}
        Let $s_a = \eta(h_a)$.
        For every $k\in[a+1,b]$, we have $\eta(p^1_{i,k}) = 3 - s_a$ and $\eta(q^1_{i,k}) = s_a$.
    \end{claim}
    \begin{proof}[Proof of the claim]
        Since $\lambda_{i,a}$ is adjacent to $\alpha_{i,a}$ and $\beta_{i,a}$ with $\eta(\alpha_{i,a})=0$ and $\eta(\beta_{i,a})=\eta(h_a) = s_a$, we have $\eta(\lambda_{i,a})=3-s_a$.
        Since $x^1_{i,a}$ is adjacent to both $\lambda_{i,a}$ and $h_a$, we have $\eta(x^1_{i,a})=0$.
        Since $y^1_{i,a}$ is adjacent to $x^1_{i,a}$ and $h_a$, we have $\eta(y^1_{i,a})=3-s_a$.
        Since $q^1_{i,a+1}$ is adjacent to $\alpha_{i,a+1}$ and $y^1_{i,a}$, we obtain $\eta(q^1_{i,a+1})=s_a$ and hence $\eta(p^1_{i,a+1})=3-s_a$.

        Suppose that $\eta(p^1_{i,k})=3-s_a$ and $\eta(q^1_{i,k})=s_a$ for some $k\in[a+1,b-1]$.
        We determine the colors of the pair in $A_{i,k+1}$ from those of the pair in $A_{i,k}$ using $x^1_{i,k},y^1_{i,k}$ in $B_{i,k}$.
        The colors in this step are as follows:
        \[
        \begin{array}{c|cc|cc|cc}
            \eta(h_k) & \eta(p^1_{i,k}) & \eta(q^1_{i,k}) & \eta(x^1_{i,k}) & \eta(y^1_{i,k}) & \eta(p^1_{i,k+1}) & \eta(q^1_{i,k+1})\\
            \hline
            s_a & 3-s_a & s_a & 0 & 3-s_a & 3-s_a & s_a\\
            3-s_a & 3-s_a & s_a & s_a & 0 & 3-s_a & s_a
        \end{array}
        \]
        If $\eta(h_k)=s_a$, then $\eta(x^1_{i,k})=0$ since $x^1_{i,k}$ is adjacent to $p^1_{i,k}$ and $h_k$.
        Since $y^1_{i,k}$ is adjacent to $x^1_{i,k}$ and $h_k$, we have $\eta(y^1_{i,k})=3-s_a$.
        Since $q^1_{i,k+1}$ is adjacent to $\alpha_{i,k+1}$ and $y^1_{i,k}$, we obtain $\eta(q^1_{i,k+1})=s_a$ and hence $\eta(p^1_{i,k+1})=3-s_a$.

        If $\eta(h_k)=3-s_a$, then $\eta(y^1_{i,k})=0$ since $y^1_{i,k}$ is adjacent to $q^1_{i,k}$ and $h_k$.
        Since $x^1_{i,k}$ is adjacent to $y^1_{i,k}$ and $h_k$, we have $\eta(x^1_{i,k})=s_a$.
        Since $p^1_{i,k+1}$ is adjacent to $\alpha_{i,k+1}$ and $x^1_{i,k}$, we obtain $\eta(p^1_{i,k+1})=3-s_a$ and hence $\eta(q^1_{i,k+1})=s_a$.
        The claim follows by induction.
        For later use, if both $(p^1_{i,k},q^1_{i,k})$ and $(p^1_{i,k+1},q^1_{i,k+1})$ have colors $(3-s_a,s_a)$, the corresponding row of the table gives a proper coloring of every edge incident with $x^1_{i,k}$ or $y^1_{i,k}$.
    \end{proof}

    We now consider the propagation from $A_{i,c+1}$ to $A_{i,b}$.

    \begin{claim}\label{app:col:claim:propagetefromctob}
        Let $s_c = \eta(h_c)$.
        For every $k \in [b, c]$, we have $\eta(p^2_{i,k})=3-s_c$ and $\eta(q^2_{i,k}) = s_c$.
    \end{claim}
    \begin{proof}[Proof of the claim]
        Since $\lambda_{i,c}$ is adjacent to $\alpha_{i,c+1}$ and $\gamma_{i,c+1}$ with $\eta(\alpha_{i,c+1})=0$ and $\eta(\gamma_{i,c+1})=\eta(h_c) = s_c$, we have $\eta(\lambda_{i,c})=3-s_c$.
        Since $x^2_{i,c}$ is adjacent to both $\lambda_{i,c}$ and $h_c$, we have $\eta(x^2_{i,c})=0$.
        Since $y^2_{i,c}$ is adjacent to $x^2_{i,c}$ and $h_c$, we have $\eta(y^2_{i,c})=3-s_c$.
        Since $q^2_{i,c}$ is adjacent to $\alpha_{i,c}$ and $y^2_{i,c}$, we obtain $\eta(q^2_{i,c})=s_c$ and hence $\eta(p^2_{i,c})=3-s_c$.

        Suppose that $\eta(p^2_{i,k+1})=3-s_c$ and $\eta(q^2_{i,k+1})=s_c$ for some $k\in[b,c-1]$.
        We determine the colors of the pair in $A_{i,k}$ from those of the pair in $A_{i,k+1}$ using $x^2_{i,k},y^2_{i,k}$ in $B_{i,k}$.
        The colors in this step are as follows:
        \[
        \begin{array}{c|cc|cc|cc}
            \eta(h_k) & \eta(p^2_{i,k+1}) & \eta(q^2_{i,k+1}) & \eta(x^2_{i,k}) & \eta(y^2_{i,k}) & \eta(p^2_{i,k}) & \eta(q^2_{i,k})\\
            \hline
            s_c & 3-s_c & s_c & 0 & 3-s_c & 3-s_c & s_c\\
            3-s_c & 3-s_c & s_c & s_c & 0 & 3-s_c & s_c
        \end{array}
        \]
        If $\eta(h_k)=s_c$, then $\eta(x^2_{i,k})=0$ since $x^2_{i,k}$ is adjacent to $p^2_{i,k+1}$ and $h_k$.
        Since $y^2_{i,k}$ is adjacent to $x^2_{i,k}$ and $h_k$, we have $\eta(y^2_{i,k})=3-s_c$.
        Since $q^2_{i,k}$ is adjacent to $\alpha_{i,k}$ and $y^2_{i,k}$, we obtain $\eta(q^2_{i,k})=s_c$ and hence $\eta(p^2_{i,k})=3-s_c$.

        If $\eta(h_k)=3-s_c$, then $\eta(y^2_{i,k})=0$ since $y^2_{i,k}$ is adjacent to $q^2_{i,k+1}$ and $h_k$.
        Since $x^2_{i,k}$ is adjacent to $y^2_{i,k}$ and $h_k$, we have $\eta(x^2_{i,k})=s_c$.
        Since $p^2_{i,k}$ is adjacent to $\alpha_{i,k}$ and $x^2_{i,k}$, we obtain $\eta(p^2_{i,k})=3-s_c$ and hence $\eta(q^2_{i,k})=s_c$.
        Starting at index $c$ and proceeding to $b$ proves the claim by induction.
        For later use, if both $(p^2_{i,k},q^2_{i,k})$ and $(p^2_{i,k+1},q^2_{i,k+1})$ have colors $(3-s_c,s_c)$, the corresponding row of the table gives a proper coloring of every edge incident with $x^2_{i,k}$ or $y^2_{i,k}$.
    \end{proof}

    By \cref{app:col:claim:propagetefromatob,app:col:claim:propagetefromctob}, we have $\eta(p^1_{i,b})=3 - \eta(h_a)$ and $\eta(p^2_{i,b})=3 - \eta(h_c)$.
    Moreover, since $\lambda_{i,b}$ is adjacent to $\alpha_{i,b}$ and $\beta_{i,b}$, we have $\eta(\lambda_{i,b}) = 3 - \eta(h_b)$.
    By \cref{app:col:claim:tww2NAE}, since the coloring of $\lambda_{i,b}$, $p^1_{i,b}$ and $p^2_{i,b}$ extends to $\nae(\lambda_{i,b},p^1_{i,b},p^2_{i,b})$, these three vertices are not all assigned the same color.
    Therefore, $h_a$, $h_b$ and $h_c$ are not all assigned the same color.

    The row $R_i$ was arbitrary, and renaming colors preserves whether $h_a,h_b,h_c$ are monochromatic.
    Thus, in the original coloring, the three vertices corresponding to every hyperedge are not all assigned the same color.
    The vertices $h_1,\ldots,h_n$ use at most two colors, as shown above.
    Relabel these colors by $1,2$ and assign each $v_k$ the color of $h_k$.
    This gives a $2$-coloring of $\mathcal{H}$.

    Conversely, suppose that $\mathcal{H}$ admits a $2$-coloring $\rho: V(\mathcal{H}) \to \{1, 2\}$.
    We construct a $3$-coloring $\eta$ of $G_\mathcal{H}$.
    For every $k\in[n]$, set $\eta(h_k) = \rho(v_k)$.
    Since distinct rows are pairwise anti-complete, it is enough to extend $\eta$ to each row separately.
    Fix $F_i = \{v_a, v_b, v_c\} \in \mathcal{F}$ with $a < b < c$ and consider the row $R_i$ corresponding to $F_i$.

    For every $k\in[n+1]$, set $\eta(\alpha_{i,k})=0$.
    For every $k\in[n]$, set $\eta(\beta_{i,k})=\eta(\gamma_{i,k+1})=\eta(h_k)$, $\eta(w_{i,k})=0$, $\eta(z_{i,k})=3-\eta(h_k)$ and $\eta(\mu_{i,k})=3-\eta(h_k)$.
    Also, for every $k\in \{a, b, c\}$, set $\eta(\lambda_{i, k}) = 3 - \eta(h_k)$.
    Every edge between vertices colored so far has endpoints of distinct colors.

    Let $s_a = \eta(h_a)$.
    For every $k\in[a+1,b]$, set $\eta(p^1_{i,k})=3-s_a$ and $\eta(q^1_{i,k})=s_a$.
    Set $\eta(x^1_{i,a})=0$ and $\eta(y^1_{i,a})=3-s_a$.
    For every $k\in[a+1,b-1]$, assign colors to $x^1_{i,k},y^1_{i,k}$ according to the table in the proof of~\cref{app:col:claim:propagetefromatob}.
    As verified in that proof, these assignments give distinct colors to the endpoints of every edge incident with $x^1_{i,k}$ or $y^1_{i,k}$ for each such $k$.
    The base case in the same proof shows that the edges incident with $x^1_{i,a}$ or $y^1_{i,a}$ also have endpoints of distinct colors.

    Let $s_c = \eta(h_c)$.
    For every $k\in[b,c]$, set $\eta(p^2_{i,k})=3-s_c$ and $\eta(q^2_{i,k})=s_c$.
    Set $\eta(x^2_{i,c})=0$ and $\eta(y^2_{i,c})=3-s_c$.
    For every $k\in[b,c-1]$, assign colors to $x^2_{i,k},y^2_{i,k}$ according to the table in the proof of~\cref{app:col:claim:propagetefromctob}.
    As verified in that proof, these assignments give distinct colors to the endpoints of every edge incident with $x^2_{i,k}$ or $y^2_{i,k}$ for each such $k$.
    The base case in the same proof shows that the edges incident with $x^2_{i,c}$ or $y^2_{i,c}$ also have endpoints of distinct colors.
    Each pair $p^r_{i,k},q^r_{i,k}$ has distinct colors $1,2$, both different from the color $0$ assigned to $\alpha_{i,k}$.
    Thus, every edge between vertices colored so far has endpoints of distinct colors.

    It remains to extend $\eta$ to $\nae(\lambda_{i,b},p^1_{i,b},p^2_{i,b})$.
    We have $\eta(\lambda_{i,b})=3-\eta(h_b)$, $\eta(p^1_{i,b})=3-\eta(h_a)$ and $\eta(p^2_{i,b})=3-\eta(h_c)$.
    Since $\rho$ is a $2$-coloring of $\mathcal{H}$, the vertices $h_a,h_b,h_c$ are not all assigned the same color.
    Therefore, $\lambda_{i,b},p^1_{i,b},p^2_{i,b}$ are not all assigned the same color.
    By \cref{app:col:claim:tww2NAE}, the coloring of $\lambda_{i,b},p^1_{i,b},p^2_{i,b}$ extends to a $3$-coloring of $\nae(\lambda_{i,b},p^1_{i,b},p^2_{i,b})$.
    This gives a $3$-coloring of the row corresponding to $F_i$.
    Since the rows are pairwise anti-complete, applying this construction to every row yields a $3$-coloring of $G_\mathcal{H}$.

    \medskip
    
\proofparagraph{A witnessing contraction sequence.}
    We next show that $\reduced{\compmaxleaf}(G_\mathcal{H}) \le 3$.
    
    \begin{claim}\label{app:col:claim:computecmlle3}
        $\reduced{\compmaxleaf}(G_\mathcal{H}) \le 3$.
    \end{claim}
    \begin{proof}[Proof of the claim]
        We construct a contraction sequence of $G_\mathcal{H}$.
        We process the rows one at a time.
        Fix $i\in[m]$ and let $F_i=\{v_a,v_b,v_c\}$ with $a<b<c$.
        The rows already processed form red paths, and the rows not yet processed have no red edges.
        During the contractions within $R_i$, every part remains contained in one block.
        In particular, every part contained in $B_{i,k}$ remains complete to $h_k$ and every part contained in an $A$-block remains anti-complete to all $h_j$.
        Hence, no red edge is incident with an $h_j$.

        For every $k\in[n]$, we first contract $w_{i,k}$ and $z_{i,k}$ and denote the resulting part by $W_{i,k}$.
        Since the adjacent vertices $w_{i,k}$ and $z_{i,k}$ have the same neighbors outside the pair (they are true twins), no red edge appears after these contractions.

        We use the following two observations to bound the max-leaf of each red component.
        If a connected graph has maximum degree at most $3$ and at most one vertex of degree $3$, every spanning tree inherits both properties and hence has at most $3$ leaves.
        A copy of $K_3$ with two pendant paths attached at distinct vertices also has max-leaf at most $3$ (see~\cref{fig:red-components}).
        In every spanning tree, only the two pendant vertices and the vertex of $K_3$ with no attached path can be leaves.
        In \textbf{Steps 1--4}, every red component is a path or a subdivision of $K_{1,3}$.
        In \textbf{Step 5}, a copy of $K_3$ with two pendant paths attached at distinct vertices also appears as a red component.

        \proofparagraph{Step 1: The propagation from $\bm{A_{i,a}}$ to $\bm{A_{i,b}}$.}
        
        We next contract the propagation from $A_{i, a}$ to $A_{i, b}$.
        We first contract $x^1_{i,a}$ and $y^1_{i,a}$ into a part $X^1_{i,a}$.
        We then consider each $k=a+1,\ldots,b-1$ in order.
        We contract $p^1_{i,k}$ and $q^1_{i,k}$ into a part $P^1_{i,k}$ and then contract $x^1_{i,k}$ and $y^1_{i,k}$ into a part $X^1_{i,k}$.
        Finally, we contract $p^1_{i,b}$ and $q^1_{i,b}$ into a part $P^1_{i,b}$.
        After these contractions, the red edges arising from this propagation form a path from $\lambda_{i,a}$ to $s^2_i$.

        When $X^1_{i,a}$ is obtained, its red neighbors are $\lambda_{i,a},p^1_{i,a+1},q^1_{i,a+1}$.
        For each $k\in[a+1,b-1]$, when $P^1_{i,k}$ is obtained, its red neighbors are $X^1_{i,k-1},x^1_{i,k},y^1_{i,k}$, and when $X^1_{i,k}$ is obtained, its red neighbors are $P^1_{i,k},p^1_{i,k+1},q^1_{i,k+1}$.
        Finally, when $P^1_{i,b}$ is obtained, its red neighbors are $X^1_{i,b-1}$ and $s^2_i$.
        Thus, at each step, the part obtained by the contraction is the only part that may have red degree $3$.

        \proofparagraph{Step 2: The propagation from $\bm{A_{i,c+1}}$ to $\bm{A_{i,b}}$.}

        We next contract the propagation from $A_{i, c+1}$ to $A_{i, b}$.
        We first contract $x^2_{i,c}$ and $y^2_{i,c}$ into a part $X^2_{i,c}$.
        We then consider each $k = c, c-1, \ldots, b+1$ in order.
        We contract $p^2_{i,k}$ and $q^2_{i,k}$ into a part $P^2_{i,k}$ and then contract $x^2_{i,k-1}$ and $y^2_{i,k-1}$ into a part $X^2_{i,k-1}$.
        Finally, we contract $p^2_{i,b}$ and $q^2_{i,b}$ into a part $P^2_{i,b}$.
        After these contractions, the red edges arising from this propagation form a path from $\lambda_{i,c}$ to $s^3_i$; the path obtained in Step~1 remains unchanged.

        When $X^2_{i,c}$ is obtained, its red neighbors are $\lambda_{i,c},p^2_{i,c},q^2_{i,c}$.
        For each $k\in[b+1,c]$, when $P^2_{i,k}$ is obtained, its red neighbors are $X^2_{i,k},x^2_{i,k-1},y^2_{i,k-1}$, and when $X^2_{i,k-1}$ is obtained, its red neighbors are $P^2_{i,k},p^2_{i,k-1},q^2_{i,k-1}$.
        Finally, when $P^2_{i,b}$ is obtained, its red neighbors are $X^2_{i,b}$ and $s^3_i$.
        Thus, at each step, the part obtained by the contraction is the only part that may have red degree $3$.

        \proofparagraph{Step 3: From $\bm{A_{i,b}}$ to $\bm{A_{i,c+1}}$.}

        We now contract $s^1_i$ and $s^2_i$ and then contract the resulting part with $s^3_i$ into a part $S_i$.
        Then, the red neighbors of $S_i$ are exactly $\lambda_{i,b}$, $P^1_{i,b}$ and $P^2_{i,b}$.
        At this point, $S_i$ is the unique part of red degree $3$.
        We next contract $\lambda_{i,b}$ and $\beta_{i,b}$ and then contract the resulting part with $P^2_{i,b}$ into a part $Q^2_{i,b}$.
        We then contract $W_{i,b}$ and $X^2_{i,b}$ into a part $R^2_{i,b}$.

        We consider each $k = b+1, \ldots, c$ in order.
        We contract $\gamma_{i,k}$ and $P^2_{i,k}$ and then contract the resulting part with $\beta_{i,k}$ into a part $Q^2_{i,k}$.
        We then contract $W_{i,k}$ and $X^2_{i,k}$ into a part $R^2_{i,k}$.
        Finally, we contract $\gamma_{i,c+1}$ and $\lambda_{i,c}$.
        If $\beta_{i,c+1}$ exists, we further contract the resulting part with $\beta_{i,c+1}$.
        We denote the final part by $Q^2_{i,c+1}$.
        After each contraction, there is at most one part of red degree $3$, while every other part has red degree at most $2$.
        After this process, the red edges form a path from $\lambda_{i,a}$ to $W_{i,c+1}$ if $\beta_{i,c+1}$ exists, and to $Q^2_{i,c+1}$ otherwise.

        \proofparagraph{Step 4: From $\bm{A_{i,b}}$ to $\bm{A_{i,a}}$.}

        We first contract $\gamma_{i,b}$ and $P^1_{i,b}$ into a part $Q^1_{i,b}$.
        We then contract $W_{i,b-1}$ and $X^1_{i,b-1}$ into a part $R^1_{i,b-1}$.

        We consider each $k=b-1,b-2,\ldots,a+1$ in order.
        We contract $\beta_{i,k}$ and $P^1_{i,k}$ and then contract the resulting part with $\gamma_{i,k}$ into a part $Q^1_{i,k}$.
        We then contract $W_{i,k-1}$ and $X^1_{i,k-1}$ into a part $R^1_{i,k-1}$.
        Finally, we contract $\lambda_{i,a}$ and $\beta_{i,a}$.
        If $\gamma_{i, a}$ exists, we further contract the resulting part with $\gamma_{i, a}$.
        We denote the final part by $Q^1_{i,a}$.

        After each contraction, there is at most one part of red degree $3$, while every other part has red degree at most $2$.
        After this process, the red edges in the processed region form a path whose left endpoint is $W_{i,a-1}$ if $\gamma_{i,a}$ exists, and $Q^1_{i,a}$ otherwise, and whose right endpoint is $W_{i,c+1}$ if $\beta_{i,c+1}$ exists, and $Q^2_{i,c+1}$ otherwise.

        \proofparagraph{Step 5: From $\bm{A_{i,1}}$ to $\bm{A_{i,n+1}}$.}

        We start with the blocks $A_{i,1},B_{i,1},\ldots,B_{i,a-1}$.
        If $a>1$, we contract $\beta_{i,1}$ and $\alpha_{i,1}$ into a part $A'_{i,1}$ and then contract $W_{i,1}$ and $\mu_{i,1}$ into a part $B'_{i,1}$.
        We consider each $k=2, \ldots, a-1$ in order.
        We contract $\gamma_{i,k}$ and $\alpha_{i,k}$, then contract the resulting part with $\beta_{i,k}$ into a part $A'_{i,k}$.
        We next contract $W_{i,k}$ and $\mu_{i,k}$ into a part $B'_{i,k}$.

        We then consider each $k=a, \ldots, b-1$ in order.
        We contract $Q^1_{i,k}$ and $\alpha_{i,k}$ into a part $A'_{i,k}$ and then contract $R^1_{i,k}$ and $\mu_{i,k}$ into a part $B'_{i,k}$.
        After each contraction, there is at most one part of red degree $3$, while every other part has red degree at most $2$.

        We now consider $A_{i,b}$.
        We first contract $Q^1_{i,b}$ and $S_i$ into a part $T_{i,b}$.
        We then contract $T_{i,b}$ and $\alpha_{i,b}$ into a part $U_{i,b}$ and finally contract $U_{i,b}$ and $Q^2_{i,b}$ into a part $A'_{i,b}$.
        When $T_{i,b}$ is obtained, the red component containing $T_{i,b}$ consists of a $K_3$ induced by $B'_{i,b-1}, T_{i,b},\alpha_{i,b}$ together with a path from $A'_{i,1}$ to $B'_{i,b-1}$ and a path from $T_{i,b}$ to $W_{i,c+1}$ if $\beta_{i,c+1}$ exists, and to $Q^2_{i,c+1}$ otherwise.
        Thus, this red component has max-leaf at most $3$.

        We next consider $B_{i,b}$.
        We first contract $R^2_{i,b}$ and $\mu_{i,b}$ into a part $B'_{i,b}$.
        We then consider each $k=b+1,\ldots,c$ in order.
        We contract $Q^2_{i,k}$ and $\alpha_{i,k}$ into a part $A'_{i,k}$ and then contract $R^2_{i,k}$ and $\mu_{i,k}$ into a part $B'_{i,k}$.
        We next contract $Q^2_{i,c+1}$ and $\alpha_{i,c+1}$ into a part $A'_{i,c+1}$.

        If $c<n$, we next consider the remaining blocks.
        We consider each $k=c+1,\ldots,n-1$ in this order.
        We first contract $W_{i,k}$ and $\mu_{i,k}$ into a part $B'_{i,k}$.
        We then contract $\alpha_{i,k+1}$ and $\gamma_{i,k+1}$, then contract the resulting part with $\beta_{i,k+1}$ into a part $A'_{i,k+1}$.
        We then contract $W_{i,n}$ and $\mu_{i,n}$ into a part $B'_{i,n}$.
        Finally, we contract $\gamma_{i,n+1}$ and $\alpha_{i,n+1}$ into a part $A'_{i,n+1}$.

        After these contractions, the parts of $R_i$ form the red path
        \[A'_{i,1}-B'_{i,1}-A'_{i,2}-B'_{i,2}-\cdots-B'_{i,n}-A'_{i,n+1}.\]
        In each of \textbf{Steps 1}, \textbf{2}, \textbf{3} and \textbf{4}, each red component has maximum degree at most $3$ and at most one degree-three vertex.
        In \textbf{Step 5}, the same holds except for the red component containing $T_{i,b}$ immediately after $Q^1_{i,b}$ and $S_i$ are contracted.
        By the two observations above, every red component appearing during these contractions has max-leaf at most $3$.

        We next contract corresponding parts of the red paths obtained from the rows.
        Starting with $R_1$, we consider the remaining rows one at a time in the order $R_2,\ldots,R_m$.
        When we consider a row $R_i$, we contract its parts with the parts obtained from the previously considered rows in the order
        \[A'_{i,1}\,,\,B'_{i,1}\,,\,A'_{i,2}\,,\,B'_{i,2}\,,\,\ldots\,,\, A'_{i,n}\,,\,B'_{i,n}\,,\,A'_{i,n+1}.\]
        At every step, each component of the red graph is a path or a subdivision of $K_{1,3}$.
        Thus, every spanning tree of such a component has at most $3$ leaves.

        For every $j\in[n]$, let $\widehat{B}_j$ be the part containing $\bigcup_{i\in[m]} B'_{i,j}$.
        Similarly, for every $j\in[n+1]$, let $\widehat{A}_j$ be the part containing $\bigcup_{i\in[m]} A'_{i,j}$.
        After these contractions, these parts form the red path
        \[\widehat{A}_1-\widehat{B}_1-\widehat{A}_2-\widehat{B}_2-\cdots-\widehat{B}_n-\widehat{A}_{n+1}.\]
        By construction, $\widehat{B}_j$ is complete to $h_j$ for every $j\in[n]$.

        Finally, for every $j\in[n]$, we contract $h_j$ with $\widehat{B}_j$.
        The remaining red graph is a path, which we contract from one endpoint.
        Therefore, every red graph in this contraction sequence has component max-leaf at most $3$.
        This contraction sequence can be constructed in polynomial time.
        This completes the proof.
    \end{proof}
    As shown above, $\mathcal{H}$ is $2$-colorable if and only if $G_{\mathcal{H}}$ is $3$-colorable.
    By \cref{app:col:claim:computecmlle3}, we have $\reduced{\compmaxleaf}(G_{\mathcal{H}})\le 3$ and both $G_{\mathcal{H}}$ and the contraction sequence can be constructed in polynomial time.
    Thus, this gives a polynomial-time reduction from \textnormal{\textsc{$3$-Uniform Hypergraph $2$-Colorability}} to \textnormal{\textsc{$3$-Coloring}} on graphs of reduced component max-leaf at most $3$, even when a witnessing contraction sequence is given.
    Since a proposed coloring can be checked on every edge in polynomial time, \textnormal{\textsc{$3$-Coloring}} is NP-complete on this class.

Every red graph in the sequence has maximum degree at most $3$, so the same sequence witnesses twin-width at most $3$.
Weak sparsity follows from \cref{ws:hardness-constructions}.
\end{proof}

\begin{samepage}
\subsection{\textsc{Feedback Vertex Set}}
\label{sec:fvs}
\label{app:fvs}

An instance of \textnormal{\textsc{Set Cover}} consists of a finite set $V$, a family $\mathcal{S}$ of subsets of $V$ and a positive integer $k$.
The problem asks whether $V$ can be covered by at most $k$ sets from $\mathcal{S}$.
It is one of Karp's NP-complete problems~\cite{Karp72}.

\fvsHardness*
\end{samepage}

\begin{proof}
    We reduce from \textnormal{\textsc{Set Cover}}.
    Let $(V, \mathcal{S}, k)$ be an instance of \textnormal{\textsc{Set Cover}} where $V = \{v_1, \ldots, v_n\}$, $\mathcal{S} = \{S_1, \ldots, S_m\}$ and $k$ is a positive integer.
    After handling trivial instances, we may assume that $n\ge1$, $m\ge2$ and each set $S_j$ is nonempty.
    We may also assume that $k\le m$.
    For each $i\in[n]$, put $Z_i=\{j\in[m]:v_i\in S_j\}$ and $d_i=|Z_i|$.
    Counting the pairs $(v_i,S_j)$ with $v_i\in S_j$ in two ways gives
    \begin{equation}\label{app:fvs:claim:sumofsetcover}
        \sum_{i\in[n]}d_i=\sum_{j\in[m]}|S_j|.
    \end{equation}

\proofparagraph{The construction.}
    We construct a graph $G_\mathcal{F}$ such that $V$ can be covered by at most $k$ sets from $\mathcal{S}$ if and only if $G_\mathcal{F}$ has a feedback vertex set of size at most $K$, where
    \[
        K=(k+1)\left(mn+\sum_{i\in[n]}d_i\right)+k
         =(k+1)\left(mn+\sum_{j\in[m]}|S_j|\right)+k.
    \]
    We introduce independent vertices $h_1,\ldots,h_m$, where $h_j$ represents the set $S_j$.
    We will show that a row for $v_i$ requires at least $m+d_i$ deletions and at least one more unless a vertex $h_j$ with $v_i\in S_j$ is selected.
    Taking $k+1$ copies of each row makes this additional cost larger than the remaining budget of $k$ for the vertices $h_j$.

    For each $v_i \in V$, we construct a row $R_i$ consisting of
    \[A_{i, 0} - B_{i, 0} - A_{i,1} - B_{i,1} - \cdots - B_{i, m}\]
    such that each $h_j$ with $j \in [m]$ is complete to $B_{i, j}$ and anti-complete to $V(R_i)\setminus B_{i, j}$.
    The vertices in the $A$-blocks are introduced in pairs with identical two-element neighborhoods.
    Retaining both vertices of such a pair and both of their neighbors leaves a $4$-cycle.
    We construct $R_i$ as follows:
    \begin{itemize}
        \item For each $j \in \{0\} \cup [m]$, we introduce two vertices $a_{i, j}$ and $b_{i, j}$ in $B_{i, j}$.
        For each $j \in [m]$ with $v_i \in S_j$, we introduce two vertices $p_{i, j}$ and $q_{i, j}$ in $B_{i, j}$ and add the edge $p_{i, j}q_{i, j}$.
        
        \item For each $j \in \{0\}\cup [m]$, we introduce two vertices $\alpha_{i, j}$ and $\overline{\alpha_{i, j}}$ in $A_{i, j}$ each with neighborhood $\{a_{i, j}, b_{i, j}\}$.
        For each $j \in [m]$ with $v_i \in S_j$, we introduce two vertices $\beta_{i, j}$ and $\overline{\beta_{i, j}}$ in $A_{i, j}$ each with neighborhood $\{a_{i, j}, p_{i, j}\}$.

        \item For each $j \in [m]$, we introduce two vertices $\gamma_{i, j}$ and $\overline{\gamma_{i, j}}$ in $A_{i, j}$ each with neighborhood $\{b_{i, j-1}, a_{i, j}\}$.
        If $v_i \in S_j$, we introduce two vertices $\delta_{i, j}$ and $\overline{\delta_{i, j}}$ in $A_{i, j}$ each with neighborhood $\{b_{i, j-1}, p_{i, j}\}$.
    \end{itemize}

    For each $i\in[n]$, we construct $k+1$ copies of $R_i$, denoted by $R_i^1,\ldots,R_i^{k+1}$.
    For each vertex $u\in V(R_i)$, let $u^\ell$ denote its copy in $R_i^\ell$.
    For each set $X\subseteq V(R_i)$, let $X^\ell=\{u^\ell:u\in X\}$.
    For every $j\in[m]$, the vertex $h_j$ is complete to $B_{i,j}^\ell$ and anti-complete to $V(R_i^\ell)\setminus B_{i,j}^\ell$.
    Distinct row copies are pairwise anti-complete, and there are no other edges.
    Let $G_\mathcal{F}$ be the resulting graph.
    The construction is illustrated in~\cref{app:fvs:fig:fvs-gadget,app:fvs:fig:fvs3-construction}.

    	\begin{figure}[tbp]
            \centering
            \begingroup
        
            \newcommand{\DrawFVSGadget}[3]{%
                \begin{scope}[xshift=#1cm]
                    \draw[block] (-.65,-2.40) rectangle (1.45,5.50);
                    \draw[block] (2.45,-2.40) rectangle (4.55,5.50);
                    \node at (.40,-2.75) {$A_{i,#2}$};
                    \node at (3.50,-2.75) {$B_{i,#2}$};
        
                    \coordinate (#2-b) at (3.50,4.80);
                    \coordinate (#2-a) at (3.50,2.45);
                    \coordinate (#2-h) at (3.50,6.50);
        
                    \coordinate (#2-alpha) at (.70,4.80);
                    \coordinate (#2-alphabar) at (.70,4.00);
                    \coordinate (#2-gamma) at (.70,2.90);
                    \coordinate (#2-gammabar) at (.70,2.10);
        
                    \draw[pair] (-.35,3.50) rectangle (1.15,5.30);
                    \draw[pair] (-.35,1.60) rectangle (1.15,3.40);
        
                    \foreach \u in {alpha,alphabar}{
                        \draw[edge] (#2-\u) -- (#2-a);
                        \draw[edge] (#2-\u) -- (#2-b);
                    }
                    \foreach \u in {gamma,gammabar}
                        \draw[edge] (#2-\u) -- (#2-a);
        
                    \draw[hubedge] (#2-h) -- (3.50,5.50);
        
                    \ifnum#3=1
                        \coordinate (#2-p) at (3.50,.05);
                        \coordinate (#2-q) at (3.50,-1.70);
                        \coordinate (#2-beta) at (.70,1.00);
                        \coordinate (#2-betabar) at (.70,.20);
                        \coordinate (#2-delta) at (.70,-.90);
                        \coordinate (#2-deltabar) at (.70,-1.70);
        
                        \draw[pair] (-.35,-.30) rectangle (1.15,1.50);
                        \draw[pair] (-.35,-2.20) rectangle (1.15,-.40);
        
                        \foreach \u in {beta,betabar}{
                            \draw[edge] (#2-\u) -- (#2-a);
                            \draw[edge] (#2-\u) -- (#2-p);
                        }
                        \foreach \u in {delta,deltabar}
                            \draw[edge] (#2-\u) -- (#2-p);
        
                        \draw[edge] (#2-p) -- (#2-q);
                    \fi
                \end{scope}
            }
        
            \newcommand{\LabelFVSGadget}[2]{%
                \foreach \u in {a,b}
                    \node[vertex,label={[lab]above:$\u$}]
                        at (#1-\u) {};
        
                \node[vertex,label={[lab]above:$h_{#1}$}]
                    at (#1-h) {};
        
                \node[vertex,label={[lab]left=4pt:$\alpha$}]
                    at (#1-alpha) {};
                \node[vertex,label={[lab]left=4pt:$\overline\alpha$}]
                    at (#1-alphabar) {};
        
                \node[vertex,label={[lab,xshift=4pt]above=2pt:$\gamma$}]
                    at (#1-gamma) {};
                \node[vertex,label={[lab,xshift=4pt]below=2pt:$\overline\gamma$}]
                    at (#1-gammabar) {};
        
                \ifnum#2=1
                    \node[vertex,label={[lab]above:$p$}]
                        at (#1-p) {};
                    \node[vertex,label={[lab]left:$q$}]
                        at (#1-q) {};
        
                    \node[vertex,label={[lab]left=4pt:$\beta$}]
                        at (#1-beta) {};
                    \node[vertex,label={[lab]left=4pt:$\overline\beta$}]
                        at (#1-betabar) {};
        
                    \node[vertex,label={[lab,xshift=4pt]above=2pt:$\delta$}]
                        at (#1-delta) {};
                    \node[vertex,label={[lab,xshift=4pt]below=2pt:$\overline\delta$}]
                        at (#1-deltabar) {};
                \fi
            }
        
            \resizebox{\linewidth}{!}{%
            \begin{tikzpicture}[
                y=.78cm,
                edge/.style={line width=.6pt},
                transition/.style={edge,black!100!black},
                hubedge/.style={black,line width=1.6pt},
                vertex/.style={
                    circle,draw=black,fill=black,
                    inner sep=1.5pt,outer sep=0pt
                },
                block/.style={
                    draw=black!55,dashed,rounded corners=4pt,
                    line width=.55pt
                },
                pair/.style={
                    draw=black!22,fill=black!2,
                    rounded corners=3pt,line width=.4pt
                },
                lab/.style={font=\small,inner sep=1pt}
            ]
                \DrawFVSGadget{0}{j}{0}
                \DrawFVSGadget{8.20}{j+1}{1}
        
                \coordinate (previous-b) at (-2.05,2.50);
        
                \foreach \u in {gamma,gammabar}
                    \draw[transition] (previous-b) -- (j-\u);
        
                \foreach \u in {gamma,gammabar,delta,deltabar}
                    \draw[transition] (j-b) -- (j+1-\u);
        
                \LabelFVSGadget{j}{0}
                \LabelFVSGadget{j+1}{1}
        
                \node[vertex,label={[lab]left:$b_{i,j-1}$}]
                    at (previous-b) {};
        
                \node at (1.75,6.10) {$v_i\notin S_j$};
                \node at (9.95,6.10) {$v_i\in S_{j+1}$};
            \end{tikzpicture}%
            }
            \endgroup
        
            \caption{Two consecutive blocks in the row $R_i$ in the construction for \textnormal{\textsc{Feedback Vertex Set}}.
            We omit $i$, $j$ and $\ell$ from some labels for readability.}
            \label{app:fvs:fig:fvs-gadget}
        \end{figure}

        \begin{figure}[tbp]
			\centering
			\begin{tikzpicture}[scale=0.85]
				\node[draw, circle, fill=black!100, inner sep=1.5pt] at (-4.25, 4.5) {};
				\node[draw, circle, fill=black!100, inner sep=1.5pt] at (-0.25, 4.5) {};
				\node[draw, circle, fill=black!100, inner sep=1.5pt] at (4.75, 4.5) {};
				\node at (-3.8, 4.5) {$h_1$};
				\node at (0.2, 4.5) {$h_2$};
				\node at (5.2, 4.5) {$h_m$};

				\draw[thick] (-10, 3) -- (1, 3);
				
				\draw[thick] (3.5, 3) -- (5, 3);
				
				\node at (2.3, 3) {{\Large $\cdots$}};
				
				\draw[very thick] (-4.25, 4.5) -- (-4.2, 3);
				\draw[very thick] (-0.25, 4.5) -- (-0.2, 3);
				\draw[very thick] (4.75, 4.5) -- (4.8, 3);
				
				\draw[thick, rounded corners, fill=blue!10] (-10, 3.5) -- (-11, 3.5) -- (-11, 2.5) -- (-9.5, 2.5) -- (-9.5, 3.5) -- (-10.1, 3.5);
				\draw[thick, rounded corners, fill=red!10] (-8, 3.5) -- (-9, 3.5) -- (-9, 2.5) -- (-7.5, 2.5) -- (-7.5, 3.5) -- (-8.1, 3.5);
                
                \draw[thick, rounded corners, fill=blue!10] (-6, 3.5) -- (-7, 3.5) -- (-7, 2.5) -- (-5.5, 2.5) -- (-5.5, 3.5) -- (-6.1, 3.5);
				\draw[thick, rounded corners, fill=red!10] (-4, 3.5) -- (-5, 3.5) -- (-5, 2.5) -- (-3.5, 2.5) -- (-3.5, 3.5) -- (-4.1, 3.5);
				\draw[thick, rounded corners, fill=blue!10] (-2, 3.5) -- (-3, 3.5) -- (-3, 2.5) -- (-1.5, 2.5) -- (-1.5, 3.5) -- (-2.1, 3.5);
				\draw[thick, rounded corners, fill=red!10] (0, 3.5) -- (-1, 3.5) -- (-1, 2.5) -- (0.5, 2.5) -- (0.5, 3.5) -- (-0.1, 3.5);
				
				\draw[thick, rounded corners, fill=red!10] (5, 3.5) -- (4, 3.5) -- (4, 2.5) -- (5.5, 2.5) -- (5.5, 3.5) -- (4.9, 3.5);
				
				\draw[very thick] (-4.25, 4.5).. controls (-4.9, 2.5).. (-4.25, 0.5);
				\draw[very thick] (-0.25, 4.5).. controls (-0.9, 2.5).. (-0.25, 0.5);
				\draw[very thick] (4.75, 4.5).. controls (4.1, 2.5).. (4.75, 0.5);
				
				\node[rotate = 90] at (-2.5, 1.75) {{\Large $\cdots$}};

				\draw[thick, rounded corners, dashed] (0, 3.7) -- (-11.2, 3.7) -- (-11.2, 2.3) -- (5.7, 2.3) -- (5.7, 3.7) -- (-0.1, 3.7);
				
				\node at (-10.25, 3) {$A_{1, 0}^1$};
				\node at (-8.25, 3) {$B_{1, 0}^1$};
                \node at (-6.25, 3) {$A_{1, 1}^1$};
				\node at (-4.25, 3) {$B_{1, 1}^1$};
				\node at (-2.25, 3) {$A_{1, 2}^1$};
				\node at (-0.25, 3) {$B_{1, 2}^1$};
				
				\node at (4.75, 3) {$B_{1, m}^1$};
				
				\node at (-11.7, 3) {$R_1^1$};
				
				\draw[thick] (-10, 0.5) -- (1, 0.5);
				
				\draw[thick] (3.5, 0.5) -- (5, 0.5);
				
				\node at (2.3, 0.5) {{\Large $\cdots$}};

                \draw[thick, rounded corners, fill=blue!10] (-10, 1) -- (-11, 1) -- (-11, 0) -- (-9.5, 0) -- (-9.5, 1) -- (-10.1, 1);
				\draw[thick, rounded corners, fill=red!10] (-8, 1) -- (-9, 1) -- (-9, 0) -- (-7.5, 0) -- (-7.5, 1) -- (-8.1, 1);
				
				\draw[thick, rounded corners, fill=blue!10] (-6, 1) -- (-7, 1) -- (-7, 0) -- (-5.5, 0) -- (-5.5, 1) -- (-6.1, 1);
				\draw[thick, rounded corners, fill=red!10] (-4, 1) -- (-5, 1) -- (-5, 0) -- (-3.5, 0) -- (-3.5, 1) -- (-4.1, 1);
				\draw[thick, rounded corners, fill=blue!10] (-2, 1) -- (-3, 1) -- (-3, 0) -- (-1.5, 0) -- (-1.5, 1) -- (-2.1, 1);
				\draw[thick, rounded corners, fill=red!10] (0, 1) -- (-1, 1) -- (-1, 0) -- (0.5, 0) -- (0.5, 1) -- (-0.1, 1);
				
				\draw[thick, rounded corners, fill=red!10] (5, 1) -- (4, 1) -- (4, 0) -- (5.5, 0) -- (5.5, 1) -- (4.9, 1);
				
				\draw[thick, rounded corners, dashed] (0, 1.2) -- (-11.2, 1.2) -- (-11.2, -0.2) -- (5.7, -0.2) -- (5.7, 1.2) -- (-0.1, 1.2);
				
				\node at (-10.25, 0.5) {$A_{1, 0}^{k+1}$};
				\node at (-8.25, 0.5) {$B_{1, 0}^{k+1}$};
                \node at (-6.25, 0.5) {$A_{1, 1}^{k+1}$};
				\node at (-4.25, 0.5) {$B_{1, 1}^{k+1}$};
				\node at (-2.25, 0.5) {$A_{1, 2}^{k+1}$};
				\node at (-0.25, 0.5) {$B_{1, 2}^{k+1}$};
				
				\node at (4.75, 0.5) {$B_{1, m}^{k+1}$};
				
				\node at (-11.7, 0.5) {$R_1^{k+1}$};
				
				\draw[very thick] (-4.25, 4.5).. controls (-5.7, 0.75).. (-4.25, -3);
				\draw[very thick] (-0.25, 4.5).. controls (-1.7, 0.75).. (-0.25, -3);
				\draw[very thick] (4.75, 4.5).. controls (3.3, 0.75).. (4.75, -3);
				
				\node[rotate = 90] at (-2.5, -1.25) {{\Large $\cdots$}};
				
				\draw[thick] (-10, 0.5-3.5) -- (1, 0.5-3.5);
				
				\draw[thick] (3.5, 0.5-3.5) -- (5, 0.5-3.5);
				
				\node at (2.3, 0.5-3.5) {{\Large $\cdots$}};

                \draw[thick, rounded corners, fill=blue!10] (-10, -2.5) -- (-11, -2.5) -- (-11, -3.5) -- (-9.5, -3.5) -- (-9.5, -2.5) -- (-10.1, -2.5);
				\draw[thick, rounded corners, fill=red!10] (-8, -2.5) -- (-9, -2.5) -- (-9, -3.5) -- (-7.5, -3.5) -- (-7.5, -2.5) -- (-8.1, -2.5);
				
				\draw[thick, rounded corners, fill=blue!10] (-6, 1-3.5) -- (-7, 1-3.5) -- (-7, 0-3.5) -- (-5.5, 0-3.5) -- (-5.5, 1-3.5) -- (-6.1, 1-3.5);
				\draw[thick, rounded corners, fill=red!10] (-4, 1-3.5) -- (-5, 1-3.5) -- (-5, 0-3.5) -- (-3.5, 0-3.5) -- (-3.5, 1-3.5) -- (-4.1, 1-3.5);
				\draw[thick, rounded corners, fill=blue!10] (-2, 1-3.5) -- (-3, 1-3.5) -- (-3, 0-3.5) -- (-1.5, 0-3.5) -- (-1.5, 1-3.5) -- (-2.1, 1-3.5);
				\draw[thick, rounded corners, fill=red!10] (0, 1-3.5) -- (-1, 1-3.5) -- (-1, 0-3.5) -- (0.5, 0-3.5) -- (0.5, 1-3.5) -- (-0.1, 1-3.5);
				
				\draw[thick, rounded corners, fill=red!10] (5, 1-3.5) -- (4, 1-3.5) -- (4, 0-3.5) -- (5.5, 0-3.5) -- (5.5, 1-3.5) -- (4.9, 1-3.5);
				
				\draw[thick, rounded corners, dashed] (0, 1.2-3.5) -- (-11.2, 1.2-3.5) -- (-11.2, -0.2-3.5) -- (5.7, -0.2-3.5) -- (5.7, 1.2-3.5) -- (-0.1, 1.2-3.5);
				
				\node at (-10.25, 0.5-3.5) {$A_{n, 0}^{k+1}$};
				\node at (-8.25, 0.5-3.5) {$B_{n, 0}^{k+1}$};
                \node at (-6.25, 0.5-3.5) {$A_{n, 1}^{k+1}$};
				\node at (-4.25, 0.5-3.5) {$B_{n, 1}^{k+1}$};
				\node at (-2.25, 0.5-3.5) {$A_{n, 2}^{k+1}$};
				\node at (-0.25, 0.5-3.5) {$B_{n, 2}^{k+1}$};
				
				\node at (4.75, 0.5-3.5) {$B_{n, m}^{k+1}$};
				
				\node at (-11.7, 0.5-3.5) {$R_n^{k+1}$};
			\end{tikzpicture}
			\caption{The construction for the \textnormal{\textsc{Feedback Vertex Set}} problem.
			Each row $R_i^\ell$ consists of $A_{i,0}^\ell, B_{i, 0}^\ell, \ldots ,A_{i,m}^\ell, B_{i,m}^\ell$.
			Each vertex $h_j$ is complete to every $B_{i,j}^\ell$ with $i\in[n]$ and $\ell \in [k+1]$ and anti-complete to the remaining vertices.}
			\label{app:fvs:fig:fvs3-construction}
		\end{figure}

\proofparagraph{Correctness.}
    For each $i\in[n]$ and $\ell\in[k+1]$, let $A_i^\ell=\bigcup_{j=0}^m A_{i,j}^\ell$ and $B_i^\ell=\bigcup_{j=0}^m B_{i,j}^\ell$.
    We first show that a feedback vertex set can be chosen to contain no vertex from the $A$-blocks.

    \begin{claim}\label{app:fvs:claim:fvs3nocontainsA}
        For every feedback vertex set $F$ of $G_\mathcal{F}$, there exists a feedback vertex set $F'$ of $G_\mathcal{F}$ such that $|F'|\le |F|$ and $F'\cap A_i^\ell=\emptyset$ for every $i \in [n]$ and $\ell \in [k+1]$.
    \end{claim}

    \begin{proof}[Proof of the claim]
        Let $w\in F\cap A_{i,j}^\ell$.
        By construction, $w$ has exactly two neighbors in $G_\mathcal{F}$, say $N_{G_\mathcal{F}}(w)=\{x,y\}$.
        If $x\in F$ or $y\in F$, then $F\setminus\{w\}$ is a feedback vertex set.
        Otherwise, let $F'=(F\setminus\{w\})\cup\{x\}$.
        The graph $G_\mathcal{F}-F'$ is obtained from the forest $(G_\mathcal{F}-F)-x$ by restoring $w$, whose only remaining neighbor is $y$.
        Thus, $G_\mathcal{F}-F'$ is a forest, and $F'$ is a feedback vertex set of the same size as $F$.
        Repeating this operation removes all vertices of the $A$-blocks from the feedback vertex set: no operation adds an $A$-vertex, and every operation strictly decreases the number of selected $A$-vertices.
    \end{proof}

    For the row bounds, let $F$ be a feedback vertex set containing no vertex from any $A_i^\ell$.

    \begin{claim}\label{app:fvs:claim:fvs3-row}
        For every $i\in[n]$ and $\ell\in[k+1]$, we have $|F\cap B_i^\ell|\ge m+d_i$.
        If $F$ contains no vertex $h_j$ such that $v_i\in S_j$, then $|F\cap B_i^\ell|\ge m+d_i+1$.
    \end{claim}

    \begin{proof}[Proof of the claim]
        Fix $i\in[n]$ and $\ell\in[k+1]$ and omit the superscript $\ell$.
        For each $j\in[m]$, let
        \[
        W_{i,j}=
        \begin{cases}
            \{b_{i,j-1},a_{i,j},p_{i,j}\}\setminus F & \text{if }v_i\in S_j,\\
            \{b_{i,j-1},a_{i,j}\}\setminus F & \text{otherwise}.
        \end{cases}
        \]
        We have $|W_{i,j}|\le1$ for every $j\in[m]$.
        Indeed, if distinct vertices $x,y$ belong to $W_{i,j}$, the construction contains a pair $w,\overline{w}$ with $w\in\{\beta_{i,j},\gamma_{i,j},\delta_{i,j}\}$ whose common neighborhood is $\{x,y\}$.
        Neither $w$ nor $\overline{w}$ belongs to $F$, so $w-x-\overline{w}-y-w$ is a cycle in $G_\mathcal{F}-F$, a contradiction.

        Since
        \[
        B_i\setminus F\subseteq\{a_{i,0},b_{i,m}\}\cup\bigcup_{j\in[m]}W_{i,j}\cup\{q_{i,j}:j\in Z_i\},
        \]
        we have $|B_i\setminus F|\le2+m+d_i$.
        The equality $|B_i|=2(m+1)+2d_i$ gives $|B_i\cap F|\ge m+d_i$.

        Now suppose that $F$ contains no vertex $h_j$ such that $v_i\in S_j$.
        If $|B_i\cap F|=m+d_i$, then $|B_i\setminus F|=2+m+d_i$.
        Equality in the preceding bound implies that $q_{i,j}\notin F$ for every $j\in Z_i$.
        For every $j\in\{0\}\cup[m]$, at most one of $a_{i,j},b_{i,j}$ is outside $F$.
        Otherwise, $\alpha_{i,j}-a_{i,j}-\overline{\alpha_{i,j}}-b_{i,j}-\alpha_{i,j}$ is a cycle in $G_\mathcal{F}-F$.
        Thus,
        \[
        |\{a_{i,j},b_{i,j}:j\in\{0\}\cup[m]\}\setminus F|\le m+1.
        \]
        At most $m+1$ vertices of $B_i\setminus F$ are among the $a_{i,j},b_{i,j}$ and at most $d_i$ are among the $q_{i,j}$.
        Since $|B_i\setminus F|=m+d_i+2$, there exists $r\in Z_i$ with $p_{i,r}\notin F$.
        We also have $q_{i,r}\notin F$, so the triangle $h_r-p_{i,r}-q_{i,r}-h_r$ implies $h_r\in F$.
        This contradicts the assumption because $r\in Z_i$ means $v_i\in S_r$.
        Hence, $|B_i\cap F|\ge m+d_i+1$, as required.
    \end{proof}

    We now prove that $V$ has a set cover of size at most $k$ if and only if $G_\mathcal{F}$ has a feedback vertex set of size at most $K$.
    Suppose first that $\mathcal{S}'\subseteq\mathcal{S}$ covers $V$ and $|\mathcal{S}'|\le k$.
    Let $H=\{h_j:S_j\in\mathcal{S}'\}$.
    For each $i\in[n]$, choose $r_i\in Z_i$ such that $S_{r_i}\in\mathcal{S}'$.
    For every $i\in[n]$ and $\ell\in[k+1]$, let
    \[
    U_i^\ell=\{a_{i,j}^\ell:0\le j<r_i\}\cup\{b_{i,j}^\ell:r_i\le j\le m\}\cup\{p_{i,r_i}^\ell\}\cup\{q_{i,j}^\ell:j\in Z_i\},
    \]
    and define $F_i^\ell=B_i^\ell\setminus U_i^\ell$.
    Since $|U_i^\ell|=m+d_i+2$ and $|B_i^\ell|=2(m+1)+2d_i$, we have $|F_i^\ell|=m+d_i$.
    Set
    \[
    F^*=H\cup\left(\bigcup_{i\in[n]}\bigcup_{\ell\in[k+1]}F_i^\ell\right).
    \]

    We show that $G_\mathcal{F}-F^*$ is a forest.
    Since $|\{a_{i,j}^\ell,b_{i,j}^\ell\}\cap U_i^\ell|=1$, each of $\alpha_{i,j}^\ell,\overline{\alpha_{i,j}^\ell}$ has at most one remaining neighbor.
    Since $|\{b_{i,j-1}^\ell,a_{i,j}^\ell\}\cap U_i^\ell|\le1$ for every $j\in[m]$, each of $\gamma_{i,j}^\ell,\overline{\gamma_{i,j}^\ell}$ has at most one remaining neighbor.
    For each $j\in[m]$, let $W_{i,j}^\ell=\{b_{i,j-1}^\ell,a_{i,j}^\ell,p_{i,j}^\ell\}$ if $v_i\in S_j$, and $W_{i,j}^\ell=\{b_{i,j-1}^\ell,a_{i,j}^\ell\}$ otherwise.
    The definition of $U_i^\ell$ gives $|W_{i,j}^\ell\cap U_i^\ell|\le1$.
    Hence, each of $w,\overline{w}$ with $w\in\{\beta_{i,j}^\ell,\delta_{i,j}^\ell\}$ introduced in the construction has at most one remaining neighbor.
    These observations account for every vertex of the $A$-blocks.

    The only edge between two vertices of $U_i^\ell$ is $p_{i,r_i}^\ell q_{i,r_i}^\ell$.
    Since $h_{r_i}\in F^*$, neither endpoint is adjacent to a remaining vertex $h_j$.
    Every other vertex of $U_i^\ell$ is adjacent to at most one remaining vertex $h_j$.
    After removing all $A$-vertices from $G_\mathcal{F}-F^*$, the graph is a disjoint union of stars centered at the remaining vertices $h_j$, isolated edges $p_{i,r_i}^\ell q_{i,r_i}^\ell$ and isolated vertices.
    The $A$-vertices are independent and each has at most one remaining neighbor, so restoring them cannot create a cycle.
    Thus, $F^*$ is a feedback vertex set, and its size satisfies
    \[
    |F^*|=|H|+\sum_{i\in[n]}\sum_{\ell\in[k+1]}|F_i^\ell|
    \le k+(k+1)\sum_{i\in[n]}(m+d_i)=K.
    \]

    Conversely, suppose that $G_\mathcal{F}$ has a feedback vertex set $F$ with $|F|\le K$.
    By \cref{app:fvs:claim:fvs3nocontainsA}, we may assume that $F$ contains no vertex from any $A_i^\ell$.
    Let
    \[
    Z=\{i\in[n]:F\text{ contains no }h_j\text{ such that }v_i\in S_j\}.
    \]
    By \cref{app:fvs:claim:fvs3-row}, every copy $R_i^\ell$ with $i\in Z$ contains at least $m+d_i+1$ vertices of $F$, while every other copy contains at least $m+d_i$ vertices of $F$.
    Summing over the $k+1$ copies of each row gives
    \[
    \begin{aligned}
    |F|&\ge(k+1)\sum_{i\in Z}(m+d_i+1)+(k+1)\sum_{i\in[n]\setminus Z}(m+d_i)\\
       &=(k+1)mn+(k+1)|Z|+(k+1)\sum_{i\in[n]}d_i\\
       &=K-k+(k+1)|Z|.
    \end{aligned}
    \]
    If $Z\ne\emptyset$, this gives $|F|\ge K+1$, contradicting $|F|\le K$.
    Hence, $Z=\emptyset$.
    The row bounds also give
    \[
    |F\cap\{h_1,\ldots,h_m\}|\le|F|-(k+1)\sum_{i\in[n]}(m+d_i)\le k.
    \]
    Let $\mathcal{S}'=\{S_j:h_j\in F\}$.
    Then, $|\mathcal{S}'|\le k$, and $Z=\emptyset$ means that every $v_i\in V$ belongs to some set in $\mathcal{S}'$.
    Thus, $\mathcal{S}'$ is a set cover of $V$ of size at most $k$.

\proofparagraph{A witnessing contraction sequence.}
    We next show that $\reduced{\compmaxleaf}(G_\mathcal{F}) \le 3$.
    \begin{claim}\label{app:fvs:claim-1}
        $\reduced{\compmaxleaf}(G_\mathcal{F}) \le 3$.
    \end{claim}

    \begin{proof}[Proof of the claim]
        We first contract within each row copy, then contract corresponding parts of different copies, and finally contract the vertices $h_j$.
        During the contractions within a row copy, each part stays in one $A$- or $B$-block.
        Thus, every edge to an $h_j$ remains black and no red edge is incident with an $h_j$.
        We verify that every red component is a path or a subdivision of $K_{1,3}$, as in~\cref{fig:red-components}.
        Every spanning tree of either form has at most $3$ leaves.

        \proofparagraph{Step 1: Contractions within the row copies.}
        For each $i\in[n]$, we contract the copies $R_i^1,\ldots,R_i^{k+1}$ in this order and then proceed to $R_{i+1}^1$.
        Thus, the rows are contracted in the order
        \[R_1^1, \ldots, R_1^{k+1}, R_2^1, \ldots, R_2^{k+1}, \ldots, R_n^1, \ldots, R_n^{k+1}.\]
        We describe the contractions for one copy $R_i^\ell$.
        Since the same contractions are applied to every copy, we omit the superscript $\ell$ from the vertices and sets.

        Fix $R_i$.
        For every $j\in \{0\}\cup [m]$, we contract $w$ and $\overline{w}$ for each $w \in \{\alpha_{i, j}, \beta_{i, j}, \gamma_{i, j}, \delta_{i, j}\}$ introduced in the construction and denote the resulting part by $P(w)$.
        Since $w$ and $\overline{w}$ have the same neighborhood in $G_\mathcal{F}$, no red edge appears after these contractions.

        We next contract each $A_{i,j}$ into one part $\mathcal{A}_{i,j}$ and each $B_{i,j}$ into one part $\mathcal{B}_{i,j}$.
        We process the block pairs in the order
        \[A_{i, m} \cup B_{i, m}, A_{i, m-1} \cup B_{i, m-1}, \ldots, A_{i, 0} \cup B_{i, 0}.\]
        Within each $A_{i,j}\cup B_{i,j}$, we only contract pairs of parts contained in the same block $A_{i,j}$ or $B_{i,j}$.

        Consider $A_{i, j} \cup B_{i,j}$.
        Suppose first that $p_{i,j}$ exists.
        The contractions within $A_{i,j}$ and $B_{i,j}$ are performed in the following order:
        \[
        \begin{aligned}
            (P(\alpha_{i,j}), P(\beta_{i,j})), \quad
            (\{p_{i,j}\}, \{q_{i,j}\})&, \quad
            (P(\gamma_{i, j}), P(\delta_{i, j})), \quad
            (\{b_{i,j}\}, \{p_{i,j}, q_{i, j}\}),\\
            (\{a_{i,j}\}, \{b_{i,j}, p_{i,j}, q_{i, j}\})&, \quad
            (P(\alpha_{i,j})\cup P(\beta_{i,j}), P(\gamma_{i, j}) \cup P(\delta_{i, j})).
        \end{aligned}
        \]
        These contractions yield
        \[\begin{aligned}
            \mathcal{A}_{i,j} &=P(\alpha_{i,j})\cup P(\beta_{i,j}) \cup P(\gamma_{i,j})\cup P(\delta_{i,j}),\\
            \mathcal{B}_{i,j} &=\{a_{i,j},b_{i,j},p_{i,j},q_{i,j}\}.
        \end{aligned}\] 

        We claim that every component of the red graph is either a path or a subdivision of $K_{1,3}$ throughout these contractions.
        In particular, we show that the part containing $b_{i,j}$ has red degree at most $3$ while every other part has red degree at most two.

        We first contract $P(\alpha_{i,j})$ and $P(\beta_{i,j})$.
        The resulting part has red edges to $\{b_{i,j}\}$ and $\{p_{i,j}\}$ and a black edge to $\{a_{i,j}\}$.
        We next contract $\{p_{i,j}\}$ and $\{q_{i,j}\}$.
        The resulting part $\{p_{i,j},q_{i,j}\}$ has red edges to $P(\alpha_{i, j}) \cup P(\beta_{i, j})$ and $P(\delta_{i,j})$.
        
        We then contract $P(\gamma_{i,j})$ and $P(\delta_{i,j})$.
        The resulting part has red edges to $\{a_{i,j}\}$ and $\{p_{i,j},q_{i,j}\}$ and a black edge to $\{b_{i,j-1}\}$.
        Thus, every part has red degree at most two in each of these contractions.

        We next contract $\{b_{i,j}\}$ and $\{p_{i,j},q_{i,j}\}$.
        The resulting part $\{b_{i,j},p_{i,j},q_{i,j}\}$ has red edges to $P(\alpha_{i, j}) \cup P(\beta_{i, j})$ and $P(\gamma_{i, j}) \cup P(\delta_{i, j})$.
        If $j<m$, it also has a red edge to $\mathcal{A}_{i,j+1}$.
        Thus, this part has red degree at most $3$ while every other part has red degree at most two.
        At $j=m$, the red edge to $\mathcal{A}_{i,j+1}$ is absent.
        Every part then has red degree at most two.

        We then contract $\{a_{i,j}\}$ and $\{b_{i,j},p_{i,j},q_{i,j}\}$.
        The resulting part $\mathcal{B}_{i,j}$ has red edges to $P(\alpha_{i, j}) \cup P(\beta_{i, j})$ and $P(\gamma_{i, j}) \cup P(\delta_{i, j})$.
        If $j<m$, it also has a red edge to $\mathcal{A}_{i,j+1}$.
        Thus, $\mathcal{B}_{i,j}$ has red degree at most $3$ while every other part has red degree at most two.

        Finally, we contract $P(\alpha_{i, j}) \cup P(\beta_{i, j})$ and $P(\gamma_{i, j}) \cup P(\delta_{i, j})$.
        The resulting part $\mathcal{A}_{i,j}$ has red edges to $\{b_{i,j-1}\}$ and $\mathcal{B}_{i,j}$.
        Moreover, $\mathcal{B}_{i,j}$ has red degree at most two.
        After all contractions inside $A_{i,j}\cup B_{i,j}$, every part has red degree at most two.
        At most one part has red degree three at any step of these contractions.

        We now consider the case where $p_{i,j}$ does not exist and $j\ge1$.
        We first contract $\{a_{i,j}\}$ and $\{b_{i,j}\}$.
        The resulting part $\mathcal{B}_{i,j}$ has a red edge to $P(\gamma_{i,j})$.
        If $j<m$, it also has a red edge to $\mathcal{A}_{i,j+1}$.
        Its edge to $P(\alpha_{i,j})$ is black.
        We then contract $P(\alpha_{i,j})$ and $P(\gamma_{i,j})$.
        The resulting part $\mathcal{A}_{i,j}$ has red edges to $\{b_{i,j-1}\}$ and $\mathcal{B}_{i,j}$.
        Thus, every part has red degree at most two.

        At $j=0$, the block $A_{i,0}$ has already become the single part $P(\alpha_{i,0})$.
        We contract $a_{i,0}$ and $b_{i,0}$ into $\mathcal{B}_{i,0}$ and put $\mathcal{A}_{i,0}=P(\alpha_{i,0})$.
        The edge between $\mathcal{A}_{i,0}$ and $\mathcal{B}_{i,0}$ is black.

        To verify that the components are paths or subdivisions of $K_{1,3}$, omit the index $i$ and write $Q$ for the path formed by the already contracted parts to the right of $b_j$ (absent at $j=m$).
        When $p_j$ is present, let $X=P(\alpha_j)\cup P(\beta_j)$, $Z=\{p_j,q_j\}$ and $Y=P(\gamma_j)\cup P(\delta_j)$.
        After the first contraction, the affected red component is the path $Q-b_j-X-p_j$.
        After the second contraction, it is the path $Q-b_j-X-Z-P(\delta_j)$.
        After the third contraction, it is the path $Q-b_j-X-Z-Y-a_j$.
        Contracting $b_j$ with $Z$ yields a path or a subdivision of $K_{1,3}$.
        The resulting part is joined to $X$, to the path $Y-a_j$ and, when $Q$ is present, to $Q$.
        Contracting it with $a_j$ shortens the path through $Y$.
        Contracting $X,Y$ together then yields the path $b_{j-1}-\mathcal{A}_{i,j}-\mathcal{B}_{i,j}-Q$.
        When $j\ge1$ and $p_j$ is absent, the affected red component after the first contraction is the path $P(\gamma_j)-\mathcal{B}_{i,j}-Q$.
        After the second contraction, it is the path $b_{j-1}-\mathcal{A}_{i,j}-\mathcal{B}_{i,j}-Q$.
        These are all the affected red edges.
        The red components in previously processed row copies remain unchanged, and no red edge is incident with a vertex $h_j$.
        This establishes acyclicity as well as the degree bounds.

        For every $j\in[m]$, the construction gives $N_{G_{\mathcal{F}}}(h_j)\cap V(R_i)=B_{i,j}$.
        Thus, every part obtained by contracting vertices in $B_{i,j}$ is complete to $h_j$.
        Therefore, the corresponding edges are black.

        We restore the superscript $\ell$ on the parts obtained in each row copy.
        After these contractions, the parts of every copy $R_i^\ell$ form the following path.
        Its first edge is black and all subsequent edges are red:
        \[\mathcal{A}_{i,0}^\ell-\mathcal{B}_{i,0}^\ell-\mathcal{A}_{i,1}^\ell-\mathcal{B}_{i,1}^\ell -\cdots-\mathcal{A}_{i,m}^\ell-\mathcal{B}_{i,m}^\ell.\]

        \proofparagraph{Step 2: Contractions across the row copies.}
        We next contract the corresponding parts of these paths.
        Starting with $R_1^1$, we consider the remaining copies one at a time in the order
        \[R_1^2, \ldots, R_1^{k+1}, R_2^1, \ldots, R_2^{k+1}, \ldots, R_n^1, \ldots, R_n^{k+1}.\]
        When we consider a copy $R_i^\ell$, we contract its parts with the parts obtained from the previously considered copies in the order
        \[\mathcal{A}_{i,0}^\ell,\mathcal{B}_{i,0}^\ell,\mathcal{A}_{i,1}^\ell,\mathcal{B}_{i,1}^\ell, \ldots,\mathcal{A}_{i,m}^\ell,\mathcal{B}_{i,m}^\ell.\]
        The first contraction merges $\mathcal{A}_{i,0}^\ell$ with the corresponding part obtained from the previously considered copies.
        Since distinct row copies are anti-complete, both edges incident with the resulting part are red, and the two paths form one red path.
        After each subsequent contraction, the affected row parts form at most three paths meeting only at the newly obtained part.
        Their union is a path or a subdivision of $K_{1,3}$.
        No other red edges are created since corresponding parts have the same neighbors among the vertices $h_j$.
        Thus, every spanning tree of a red component has at most $3$ leaves.
        
        \proofparagraph{Step 3: Remaining contractions.}
        For every $j\in[m]$, let $\widehat{B}_j$ be the part containing
        \[
            \bigcup_{i\in[n]}\bigcup_{\ell\in[k+1]}\mathcal{B}_{i,j}^\ell.
        \]
        By construction, $\widehat{B}_j$ is complete to $h_j$.

        Finally, we contract $h_j$ with $\widehat{B}_j$ for every $j\in[m]$.
        The remaining red graph is a path.
        We contract this path from one endpoint.
        Therefore, every red graph in this contraction sequence has component max-leaf at most $3$.
        This contraction sequence can be constructed in polynomial time.
        This completes the proof.
    \end{proof}
    As shown above, $V$ has a set cover of size at most $k$ if and only if $G_\mathcal{F}$ has a feedback vertex set of size at most $K$.
    Moreover, the above contraction sequence witnesses that $\reduced{\compmaxleaf}(G_\mathcal{F})\le 3$.
    Both $G_\mathcal{F}$ and the contraction sequence can be constructed in polynomial time.
    Thus, this gives a polynomial-time reduction from \textnormal{\textsc{Set Cover}} to \textnormal{\textsc{Feedback Vertex Set}} on graphs of reduced component max-leaf at most $3$, even when a witnessing contraction sequence is given.
    Since one can check in polynomial time whether a given vertex set is a feedback vertex set of the required size, \textnormal{\textsc{Feedback Vertex Set}} is NP-complete on this class.

Every red graph in the sequence has maximum degree at most $3$, so the same sequence witnesses twin-width at most $3$.
Weak sparsity follows from \cref{ws:hardness-constructions}.
\end{proof}

\cref{prop:planar-exclusion} now gives the promised hardness under the exclusion of a fixed planar induced minor.

\planarHardness*
\begin{proof}
Apply \cref{prop:planar-exclusion} to the two reductions.
\end{proof}

\section{Hardness of \textsc{Dominating Set}}
\label{sec:domination}
\label{app:ds}
\label{app:domination}

\dsHardness*

    \begin{proof}
        We reduce from \textnormal{\textsc{Set Cover}}.
        Let $\mathcal{C} = (V, \mathcal{S}, k)$ be an instance of \textnormal{\textsc{Set Cover}} where $V = \{v_1, \ldots, v_n\}$, $\mathcal{S} = \{S_1, \ldots, S_m\}$ and $k$ is a positive integer.
        After handling trivial instances, we may assume that $n,m\ge1$ and each set $S_j$ is nonempty.
        We may also assume that $k\le m$.

\proofparagraph{The construction.}
        We construct a graph $G_{\mathcal{C}}$ such that $V$ can be covered by at most $k$ sets from $\mathcal{S}$ if and only if $G_{\mathcal{C}}$ has a dominating set of size at most $K$, where
        \[
            K=2(k+1)mn+k.
        \]
        We introduce independent vertices $h_1,\ldots,h_m$, where $h_j$ represents the set $S_j$.
        For any fixed choice of these vertices, we will show that the minimum number of additional vertices needed to dominate a row for $v_i$ is $2m$ if the chosen sets cover $v_i$ and $2m+1$ otherwise.
        We take $k+1$ copies of each row so that leaving one element uncovered requires more than the remaining budget of $k$ vertices.
        For every $v_i \in V$, we construct a row $R_i$ consisting of
        \[A_{i,1} - B_{i,1} - A_{i,2} - B_{i,2} - \cdots - B_{i, m} - A_{i, m+1}\]
        such that each $h_j$ with $j \in [m]$ is complete to $B_{i, j}$ and anti-complete to $V(R_i)\setminus B_{i, j}$.
        The optional edge $b_{i,j}x_{i,j}$ records whether $v_i\in S_j$, and the edges $q_{i,j}d_{i,j+1}$ connect consecutive blocks.
        We construct $R_i$ as follows:
        \begin{itemize}
            \item For every $j \in [m]$, let $A_{i, j} = \{a_{i, j}, b_{i, j}, c_{i, j}, d_{i, j}\}$ be an independent set and let $A_{i, m+1} = \{d_{i, m+1}\}$.
            \item For every $j \in [m]$, let $B_{i, j} = \{x_{i, j}, y_{i, j}, z_{i, j}, p_{i, j}, q_{i, j}\}$ with edge set $\{x_{i, j} y_{i, j},\, y_{i, j} z_{i, j},\, p_{i, j} q_{i, j}\}$.
            \item For every $j \in [m]$, we add the edges $a_{i, j} x_{i, j}$, $a_{i, j} y_{i, j}$, $b_{i, j} y_{i, j}$, $c_{i, j} p_{i, j}$, $c_{i, j} q_{i, j}$, $d_{i, j} x_{i, j}$, $d_{i, j} p_{i, j}$ and the edge $q_{i, j} d_{i, j+1}$.
            \item For every $j \in [m]$ with $v_i \in S_j$, we add the edge $b_{i,j}x_{i,j}$.
        \end{itemize}

        The construction is illustrated in \cref{app:ds:fig:dominating-set-block-construction}.
        For every $i\in[n]$, we construct $k+1$ copies of $R_i$, denoted by $R_i^1, \ldots, R_i^{k+1}$ (see \cref{app:ds:fig:dominatingset-construction}).
        For every vertex $v\in V(R_i)$, let $v^\ell$ denote its copy in $R_i^\ell$.
        For every set $X\subseteq V(R_i)$, let $X^\ell=\{v^\ell:v\in X\}$.
        For every $j\in[m]$, the vertex $h_j$ is complete to $B_{i,j}^\ell$ and anti-complete to $V(R_i^\ell)\setminus B_{i,j}^\ell$.
        Distinct row copies are pairwise anti-complete, and there are no other edges.
        Let $G_{\mathcal{C}}$ be the resulting graph.
        \begin{figure}[tbp]
            \centering
            \begingroup
            \newcommand{\DrawDSGadget}[4]{%
                \begin{scope}[xshift=#1cm]
                    \foreach \xa/\ya/\xb/\yb in {
                        -4/1.5/-2/0.5, -4/1.5/-1/1,
                        -4/0.5/-1/1,
                        -4/-0.5/-1.7/-1, -4/-0.5/-0.3/-1,
                        -4/-1.5/-1.7/-1, -4/-1.5/-2/0.5
                    }{
                        \draw[thick] (\xa,\ya) -- (\xb,\yb);
                    }
                    \ifnum#3=1
                    \draw[thick,red] (-4,0.5) -- (-2,0.5);
                    \fi
                    \foreach \xa/\ya/\xb/\yb in {
                        -2/0.5/-1/1, -1/1/0/0.5, -1.7/-1/-0.3/-1
                    }{
                        \draw[thick] (\xa,\ya) -- (\xb,\yb);
                    }
                    \foreach \y in {1.5,0.5,-0.5,-1.5}{
                        \node[draw,circle,fill=black!100,inner sep=1.5pt]
                        at (-4,\y) {};
                    }
                    \foreach \x/\y in {-2/0.5,-1/1,0/0.5,-1.7/-1,-0.3/-1}{
                        \draw[thick] (-1,2.7) -- (\x,\y);
                    }
                    \node[draw,circle,fill=black!100,inner sep=1.5pt]
                    at (-1,2.7) {};
                    \node at (#4,2.7) {$h_{#2}$};
                    \foreach \x/\y in {-2/0.5,-1/1,0/0.5,-1.7/-1,-0.3/-1}{
                        \node[draw,circle,fill=black!100,inner sep=1.5pt]
                        at (\x,\y) {};
                    }
                \end{scope}
            }
            \newcommand{\DrawDSBoxes}[1]{%
                \begin{scope}[xshift=#1cm]
                    \draw[rounded corners,dashed,thick]
                    (-4,2) -- (-4.5,2) -- (-4.5,-2)
                    -- (-3.5,-2) -- (-3.5,2) -- (-4.1,2);
                    \draw[rounded corners,dashed,thick]
                    (-1,2) -- (-2.5,2) -- (-2.5,-2)
                    -- (0.5,-2) -- (0.5,2) -- (-1.1,2);
                \end{scope}
            }

            \newcommand{\LabelDSVertices}[1]{%
                \begin{scope}[xshift=#1cm]
                    \foreach \name/\y in {a/1.5,b/0.5,c/-0.5}{
                        \node[font=\small,anchor=east,inner sep=1pt,xshift=-3pt]
                        at (-4,\y) {$\name$};
                    }
                    \node[font=\small,anchor=north,inner sep=1pt,yshift=-4pt]
                    at (-4,-1.5) {$d$};
                    \foreach \name/\x/\y in {
                        x/-2/0.5, y/-1/1, z/0/0.5,
                        p/-1.7/-1, q/-0.3/-1
                    }{
                        \node[font=\small,anchor=north,inner sep=1pt,yshift=-4pt]
                        at (\x,\y) {$\name$};
                    }
                \end{scope}
            }

            \begin{tikzpicture}
                \draw[thick] (-0.3,-1) -- (2,-1.5);
                \DrawDSGadget{0}{j}{0}{-0.6}
                \DrawDSGadget{6}{j+1}{1}{-0.4}
                \DrawDSBoxes{0}
                \DrawDSBoxes{6}
                \LabelDSVertices{0}
                \LabelDSVertices{6}
                \node at (-2.5,2.5) {$v_i\notin S_j$};
                \node at (3.5,2.5) {$v_i\in S_{j+1}$};
                \node at (-4,-2.4) {$A_{i,j}$};
                \node at (-1,-2.4) {$B_{i,j}$};
                \node at (2,-2.4) {$A_{i,j+1}$};
                \node at (5,-2.4) {$B_{i,j+1}$};
            \end{tikzpicture}
            \endgroup

            \caption{Two consecutive blocks in the row $R_i$ in the construction for \textnormal{\textsc{Dominating Set}}.
            We omit $i$, $j$ and $\ell$ from some vertex labels for readability.}
            \label{app:ds:fig:dominating-set-block-construction}
        \end{figure}

        \begin{figure}
            \centering
            \begin{tikzpicture}[scale=0.9]
                \node[draw, circle, fill=black!100, inner sep=1.5pt] at (-4.25, 4.5) {};
                \node[draw, circle, fill=black!100, inner sep=1.5pt] at (-0.25, 4.5) {};
                \node[draw, circle, fill=black!100, inner sep=1.5pt] at (4.75, 4.5) {};
                \node at (-3.8, 4.5) {$h_1$};
                \node at (0.2, 4.5) {$h_2$};
                \node at (5.2, 4.5) {$h_m$};

                \draw[thick] (-6, 3) -- (-4, 3);
                \draw[thick] (-4, 3) -- (-2, 3);
                \draw[thick] (-2, 3) -- (0, 3);
                \draw[thick] (0, 3) -- (1, 3);

                \draw[thick] (3.5, 3) -- (5, 3);
                \draw[thick] (5, 3) -- (7, 3);

                \node at (2.3, 3) {{\Large $\cdots$}};

                \draw[very thick] (-4.25, 4.5) -- (-4.2, 3);
                \draw[very thick] (-0.25, 4.5) -- (-0.2, 3);
                \draw[very thick] (4.75, 4.5) -- (4.8, 3);

                \draw[thick, rounded corners, fill=blue!10] (-6, 3.5) -- (-7, 3.5) -- (-7, 2.5) -- (-5.5, 2.5) -- (-5.5, 3.5) -- (-6.1, 3.5);
                \draw[thick, rounded corners, fill=red!10] (-4, 3.5) -- (-5, 3.5) -- (-5, 2.5) -- (-3.5, 2.5) -- (-3.5, 3.5) -- (-4.1, 3.5);
                \draw[thick, rounded corners, fill=blue!10] (-2, 3.5) -- (-3, 3.5) -- (-3, 2.5) -- (-1.5, 2.5) -- (-1.5, 3.5) -- (-2.1, 3.5);
                \draw[thick, rounded corners, fill=red!10] (0, 3.5) -- (-1, 3.5) -- (-1, 2.5) -- (0.5, 2.5) -- (0.5, 3.5) -- (-0.1, 3.5);

                \draw[thick, rounded corners, fill=red!10] (5, 3.5) -- (4, 3.5) -- (4, 2.5) -- (5.5, 2.5) -- (5.5, 3.5) -- (4.9, 3.5);
                \draw[thick, rounded corners, fill=blue!10] (7, 3.5) -- (6, 3.5) -- (6, 2.5) -- (7.5, 2.5) -- (7.5, 3.5) -- (6.9, 3.5);

                \draw[very thick] (-4.25, 4.5).. controls (-4.9, 2.5).. (-4.25, 0.5);
                \draw[very thick] (-0.25, 4.5).. controls (-0.9, 2.5).. (-0.25, 0.5);
                \draw[very thick] (4.75, 4.5).. controls (4.1, 2.5).. (4.75, 0.5);

                \node[rotate = 90] at (0.5, 1.75) {{\Large $\cdots$}};

                \draw[thick, rounded corners, dashed] (0, 3.7) -- (-7.2, 3.7) -- (-7.2, 2.3) -- (7.7, 2.3) -- (7.7, 3.7) -- (-0.1, 3.7);

                \node at (-6.25, 3) {$A_{1, 1}^1$};
                \node at (-4.25, 3) {$B_{1, 1}^1$};
                \node at (-2.25, 3) {$A_{1, 2}^1$};
                \node at (-0.25, 3) {$B_{1, 2}^1$};

                \node at (4.75, 3) {$B_{1, m}^1$};
                \node at (6.75, 3) {$A_{1, m+1}^1$};

                \node at (-7.7, 3) {$R_1^1$};

                \draw[thick] (-6, 0.5) -- (-4, 0.5);
                \draw[thick] (-4, 0.5) -- (-2, 0.5);
                \draw[thick] (-2, 0.5) -- (0, 0.5);
                \draw[thick] (0, 0.5) -- (1, 0.5);

                \draw[thick] (3.5, 0.5) -- (5, 0.5);
                \draw[thick] (5, 0.5) -- (7, 0.5);

                \node at (2.3, 0.5) {{\Large $\cdots$}};

                \draw[thick, rounded corners, fill=blue!10] (-6, 1) -- (-7, 1) -- (-7, 0) -- (-5.5, 0) -- (-5.5, 1) -- (-6.1, 1);
                \draw[thick, rounded corners, fill=red!10] (-4, 1) -- (-5, 1) -- (-5, 0) -- (-3.5, 0) -- (-3.5, 1) -- (-4.1, 1);
                \draw[thick, rounded corners, fill=blue!10] (-2, 1) -- (-3, 1) -- (-3, 0) -- (-1.5, 0) -- (-1.5, 1) -- (-2.1, 1);
                \draw[thick, rounded corners, fill=red!10] (0, 1) -- (-1, 1) -- (-1, 0) -- (0.5, 0) -- (0.5, 1) -- (-0.1, 1);

                \draw[thick, rounded corners, fill=red!10] (5, 1) -- (4, 1) -- (4, 0) -- (5.5, 0) -- (5.5, 1) -- (4.9, 1);
                \draw[thick, rounded corners, fill=blue!10] (7, 1) -- (6, 1) -- (6, 0) -- (7.5, 0) -- (7.5, 1) -- (6.9, 1);

                \draw[thick, rounded corners, dashed] (0, 1.2) -- (-7.2, 1.2) -- (-7.2, -0.2) -- (7.7, -0.2) -- (7.7, 1.2) -- (-0.1, 1.2);

                \node at (-6.25, 0.5) {$A_{1, 1}^{k+1}$};
                \node at (-4.25, 0.5) {$B_{1, 1}^{k+1}$};
                \node at (-2.25, 0.5) {$A_{1, 2}^{k+1}$};
                \node at (-0.25, 0.5) {$B_{1, 2}^{k+1}$};

                \node at (4.75, 0.5) {$B_{1, m}^{k+1}$};
                \node at (6.75, 0.5) {$A_{1, m+1}^{k+1}$};

                \node at (-7.7, 0.5) {$R_1^{k+1}$};

                \draw[very thick] (-4.25, 4.5).. controls (-5.7, 0.75).. (-4.25, -3);
                \draw[very thick] (-0.25, 4.5).. controls (-1.7, 0.75).. (-0.25, -3);
                \draw[very thick] (4.75, 4.5).. controls (3.3, 0.75).. (4.75, -3);

                \node[rotate = 90] at (0.5, -1.25) {{\Large $\cdots$}};

                \draw[thick] (-6, 0.5-3.5) -- (-4, 0.5-3.5);
                \draw[thick] (-4, 0.5-3.5) -- (-2, 0.5-3.5);
                \draw[thick] (-2, 0.5-3.5) -- (0, 0.5-3.5);
                \draw[thick] (0, 0.5-3.5) -- (1, 0.5-3.5);

                \draw[thick] (3.5, 0.5-3.5) -- (5, 0.5-3.5);
                \draw[thick] (5, 0.5-3.5) -- (7, 0.5-3.5);

                \node at (2.3, 0.5-3.5) {{\Large $\cdots$}};

                \draw[thick, rounded corners, fill=blue!10] (-6, 1-3.5) -- (-7, 1-3.5) -- (-7, 0-3.5) -- (-5.5, 0-3.5) -- (-5.5, 1-3.5) -- (-6.1, 1-3.5);
                \draw[thick, rounded corners, fill=red!10] (-4, 1-3.5) -- (-5, 1-3.5) -- (-5, 0-3.5) -- (-3.5, 0-3.5) -- (-3.5, 1-3.5) -- (-4.1, 1-3.5);
                \draw[thick, rounded corners, fill=blue!10] (-2, 1-3.5) -- (-3, 1-3.5) -- (-3, 0-3.5) -- (-1.5, 0-3.5) -- (-1.5, 1-3.5) -- (-2.1, 1-3.5);
                \draw[thick, rounded corners, fill=red!10] (0, 1-3.5) -- (-1, 1-3.5) -- (-1, 0-3.5) -- (0.5, 0-3.5) -- (0.5, 1-3.5) -- (-0.1, 1-3.5);

                \draw[thick, rounded corners, fill=red!10] (5, 1-3.5) -- (4, 1-3.5) -- (4, 0-3.5) -- (5.5, 0-3.5) -- (5.5, 1-3.5) -- (4.9, 1-3.5);
                \draw[thick, rounded corners, fill=blue!10] (7, 1-3.5) -- (6, 1-3.5) -- (6, 0-3.5) -- (7.5, 0-3.5) -- (7.5, 1-3.5) -- (6.9, 1-3.5);

                \draw[thick, rounded corners, dashed] (0, 1.2-3.5) -- (-7.2, 1.2-3.5) -- (-7.2, -0.2-3.5) -- (7.7, -0.2-3.5) -- (7.7, 1.2-3.5) -- (-0.1, 1.2-3.5);

                \node at (-6.25, 0.5-3.5) {$A_{n, 1}^{k+1}$};
                \node at (-4.25, 0.5-3.5) {$B_{n, 1}^{k+1}$};
                \node at (-2.25, 0.5-3.5) {$A_{n, 2}^{k+1}$};
                \node at (-0.25, 0.5-3.5) {$B_{n, 2}^{k+1}$};

                \node at (4.75, 0.5-3.5) {$B_{n, m}^{k+1}$};
                \node at (6.75, 0.5-3.5) {$A_{n, m+1}^{k+1}$};

                \node at (-7.7, 0.5-3.5) {$R_n^{k+1}$};
            \end{tikzpicture}
            \caption{The \textnormal{\textsc{Dominating Set}} construction.
            Each row copy alternates $A$- and $B$-blocks and ends with its final $A$-block.
            Each vertex $h_j$ is complete to every $B_{i,j}^{\ell}$ and anti-complete to the remaining vertices.}
            \label{app:ds:fig:dominatingset-construction}
        \end{figure}

\proofparagraph{Correctness.}
        Let $D$ be a minimum dominating set of $G_{\mathcal{C}}$.
        We first determine the number of vertices of $D$ in each row copy.

        \begin{claim}\label{app:ds:claim:rowdomset}
            For every $i\in[n]$ and $\ell\in[k+1]$, if $D$ contains a vertex $h_j$ such that $v_i\in S_j$, then $|D\cap V(R_i^\ell)|=2m$; otherwise, $|D\cap V(R_i^\ell)| = 2m+1$.
        \end{claim}
        \begin{proof}[Proof of the claim]
            Fix $i\in[n]$ and $\ell\in[k+1]$ and omit the superscript $\ell$.
            For every $r\in[m]$, consider
            \[
            N_{R_i}[a_{i,r}]=\{a_{i,r},x_{i,r},y_{i,r}\},\qquad
            N_{R_i}[c_{i,r}]=\{c_{i,r},p_{i,r},q_{i,r}\}.
            \]
            These $2m$ sets are pairwise disjoint, and no vertex outside $R_i$ is adjacent to any $a_{i,r}$ or $c_{i,r}$.
            The set $D$ must meet each of them, so $|D\cap V(R_i)|\ge2m$.

            The only vertices outside $R_i$ with neighbors in $R_i$ are $h_1,\ldots,h_m$.
            Any set of vertices in $R_i$ that meets every $B_{i,r}$ dominates all these vertices because $h_r$ is complete to $B_{i,r}$.
            Thus, replacing $D\cap V(R_i)$ by such a set preserves domination outside $R_i$.
            If the replacement also dominates $R_i$ together with $D\setminus V(R_i)$, it gives a dominating set of $G_{\mathcal{C}}$.

            Suppose first that $D$ contains a vertex $h_j$ such that $v_i\in S_j$.
            Consider
            \[
            D_i=\{y_{i,r},p_{i,r}:r\in[j-1]\}\cup\{x_{i,j},q_{i,j}\}\cup\{y_{i,r},q_{i,r}:r\in[m]\setminus[j]\}.
            \]
            We have $|D_i|=2m$.
            For every $r\in[j-1]$, the vertices $y_{i,r},p_{i,r}$ dominate $A_{i,r}\cup B_{i,r}$.
            Since $v_i\in S_j$, the edge $b_{i,j}x_{i,j}$ exists, and $x_{i,j},q_{i,j}$ together with $h_j$ dominate $A_{i,j}\cup B_{i,j}$.
            For every $r\in[m]\setminus[j]$, the vertices $y_{i,r},q_{i,r}$ dominate $A_{i,r}\cup B_{i,r}$ except $d_{i,r}$.
            Each such $d_{i,r}$ is dominated by $q_{i,r-1}\in D_i$.
            The final vertex $d_{i,m+1}$ is dominated by $q_{i,m}\in D_i$.
            Thus, $D_i\cup\{h_j\}$ dominates $R_i$.
            Since $D_i$ meets every $B_{i,r}$ and $h_j\in D$, the replacement observation shows that $(D\setminus V(R_i))\cup D_i$ is a dominating set of $G_{\mathcal{C}}$.
            Minimality of $D$ gives $|D\cap V(R_i)|\le|D_i|=2m$.
            Together with the lower bound, this proves $|D\cap V(R_i)|=2m$.

            Now suppose that $D$ contains no vertex $h_j$ such that $v_i\in S_j$.
            For every $r\in[m]$, define
            \[
            Z_r=
            \begin{cases}
                \{y_{i,r},z_{i,r}\},&\text{if }h_r\notin D,\\
                \{b_{i,r},y_{i,r}\},&\text{if }h_r\in D.
            \end{cases}
            \]
            If $h_r\notin D$, then $D$ contains $y_{i,r}$ or $z_{i,r}$ to dominate $z_{i,r}$.
            If $h_r\in D$, then $v_i\notin S_r$, so the edge $b_{i,r}x_{i,r}$ is absent.
            In this case, $D$ contains $b_{i,r}$ or $y_{i,r}$ to dominate $b_{i,r}$.
            Hence, $D$ meets $Z_r$ for every $r\in[m]$.

            The $m$ sets $Z_r$ and the $m+1$ closed neighborhoods $N_{R_i}[d_{i,s}]$, $s\in[m+1]$, are pairwise disjoint.
            No vertex outside $R_i$ is adjacent to any $d_{i,s}$, so $D$ must also meet each of these closed neighborhoods.
            It follows that $|D\cap V(R_i)|\ge m+(m+1)=2m+1$.
            This bound is attained inside $R_i$ by
            \[
            D_i'=\{d_{i,1}\}\cup\{y_{i,r},q_{i,r}:r\in[m]\}.
            \]
            For every $r\in[m]$, the vertices $y_{i,r},q_{i,r}$ dominate $A_{i,r}\cup B_{i,r}$ except $d_{i,r}$.
            The vertex $d_{i,1}$ is selected, and every $d_{i,s}$ with $s\in[2,m+1]$ is dominated by $q_{i,s-1}$.
            Thus, $D_i'$ dominates $R_i$ and has size $2m+1$.
            Since $D_i'$ meets every $B_{i,r}$, the same replacement observation shows that $(D\setminus V(R_i))\cup D_i'$ is a dominating set of $G_{\mathcal{C}}$.
            Minimality of $D$ gives $|D\cap V(R_i)|\le2m+1$.
            Together with the lower bound, this proves $|D\cap V(R_i)|=2m+1$.
        \end{proof}

        We now prove that $V$ has a set cover of size at most $k$ if and only if $G_{\mathcal{C}}$ has a dominating set of size at most $K$.
        Suppose first that $\mathcal{S}'\subseteq\mathcal{S}$ covers $V$ and $|\mathcal{S}'|\le k$.
        Let $H=\{h_j:S_j\in\mathcal{S}'\}$.
        For every $i\in[n]$, choose $j_i\in[m]$ such that $S_{j_i}\in\mathcal{S}'$ and $v_i\in S_{j_i}$.
        For every $i\in[n]$ and $\ell\in[k+1]$, let $D_i^\ell$ be the copy in $R_i^\ell$ of the displayed set $D_i$ with $j=j_i$.
        Then, $|D_i^\ell|=2m$ and $D_i^\ell\cup\{h_{j_i}\}$ dominates $R_i^\ell$.
        Each $D_i^\ell$ also meets every $B_{i,r}^\ell$ and hence dominates every $h_r$.
        Thus, $H\cup\left(\bigcup_{i\in[n]}\bigcup_{\ell\in[k+1]}D_i^\ell\right)$ is a dominating set of $G_{\mathcal{C}}$ of size at most
        \[
        |H|+2m(k+1)n\le k+2(k+1)mn=K.
        \]

        Conversely, suppose that $G_{\mathcal{C}}$ has a dominating set of size at most $K$.
        The minimum dominating set $D$ considered above then satisfies $|D|\le K$.
        Let
        \[
        z=\left|\{i\in[n]:D\text{ contains no }h_j\text{ such that }v_i\in S_j\}\right|.
        \]
        By \cref{app:ds:claim:rowdomset}, every copy $R_i^\ell$ for one of these $z$ elements contains exactly $2m+1$ vertices of $D$, while every other copy contains exactly $2m$ vertices of $D$.
        Summing over the $k+1$ copies of each row gives
        \[
        (2m+1)(k+1)z+2m(k+1)(n-z)=2(k+1)mn+(k+1)z\le|D|\le K.
        \]
        If $z\ge1$, this gives $|D|\ge2(k+1)mn+k+1=K+1$, a contradiction.
        Hence, $z=0$.
        The row copies contain exactly $2(k+1)mn$ vertices of $D$, so
        \[
        |D\cap\{h_1,\ldots,h_m\}|=|D|-2(k+1)mn\le k.
        \]
        Let $\mathcal{S}'=\{S_j:h_j\in D\}$.
        Then, $|\mathcal{S}'|\le k$, and $z=0$ means that every $v_i\in V$ belongs to some set in $\mathcal{S}'$.
        Thus, $\mathcal{S}'$ is a set cover of $V$ of size at most $k$.

\proofparagraph{A witnessing contraction sequence.}
        It remains to show that $\reduced{\compmaxleaf}(G_\mathcal{C}) \le 3$.
        \begin{claim}\label{app:ds:claim-1}
            $\reduced{\compmaxleaf}(G_\mathcal{C}) \le 3$.
        \end{claim}
        \begin{proof}[Proof of the claim]
            We first contract within every row copy, then contract corresponding parts of different copies, and finally contract the vertices $h_j$.
            During the first two stages, the parts containing vertices from $B_{i,j}^\ell$ remain complete to $h_j$ and the parts containing vertices from the $A$-blocks remain anti-complete to all $h_j$.
            Hence, no red edge is incident with an $h_j$.
            We verify that every red component is a path or a subdivision of $K_{1,3}$, as in~\cref{fig:red-components}.
            Every spanning tree of either form has at most $3$ leaves.
            For every $i\in[n]$, we contract the copies $R_i^1,\ldots,R_i^{k+1}$ in this order and then proceed to $R_{i+1}^1$.
            Thus, the rows are contracted in the order
            \[R_1^1, \ldots, R_1^{k+1}, R_2^1, \ldots, R_2^{k+1}, \ldots, R_n^1, \ldots, R_n^{k+1}.\]
            We describe the contractions for a fixed copy $R_i^\ell$.
            We apply the same contractions to every row copy, so we omit the superscript $\ell$ from the vertex and set notation.

            \proofparagraph{Step 1: Contractions within the row copies.}

            Fix $R_i$.
            We contract each $A_{i,j}$ into one part $\mathcal{A}_{i,j}$ and each $B_{i,j}$ into one part $\mathcal{B}_{i,j}$.
            We perform these contractions in the order
            \[A_{i, 1} \cup B_{i, 1}, A_{i, 2} \cup B_{i, 2}, \ldots, A_{i, m} \cup B_{i, m}.\]
            Within each $A_{i,j}\cup B_{i,j}$, we only contract parts contained in the same block $A_{i,j}$ or $B_{i,j}$.

            Consider $A_{i, j} \cup B_{i,j}$.
            The contractions within $A_{i,j}$ and $B_{i,j}$ are performed in the following order:
            \[
            \begin{aligned}
                &(\{a_{i,j}\}, \{b_{i,j}\}), \quad
                (\{x_{i,j}\}, \{z_{i,j}\}), \quad
                (\{x_{i,j},z_{i,j}\}, \{y_{i,j}\}), \quad
                (\{c_{i,j}\}, \{d_{i,j}\}),\\
                &(\{p_{i,j}\}, \{q_{i,j}\}), \quad
                (\{a_{i,j},b_{i,j}\}, \{c_{i,j},d_{i,j}\}), \quad
                (\{x_{i,j},y_{i,j},z_{i,j}\}, \{p_{i,j},q_{i,j}\}).
            \end{aligned}
            \]
            After these contractions, we obtain $\mathcal{A}_{i, j} = \{a_{i, j}, b_{i, j}, c_{i, j}, d_{i, j}\}$ and $\mathcal{B}_{i, j} = \{x_{i, j}, y_{i, j}, z_{i, j}, p_{i, j}, q_{i, j}\}$.

            We show that every red component is a path or a subdivision of $K_{1,3}$ throughout these contractions.
            In particular, every connected component of the red graph has max-leaf at most $3$.

            Just before we process $A_{i,j}\cup B_{i,j}$, the blocks with indices less than $j$ have already been contracted.
            If $j > 1$, then the already contracted part of $R_i$ forms the red path
            \[\mathcal{A}_{i,1}-\mathcal{B}_{i,1}-\cdots-\mathcal{A}_{i,j-1}-\mathcal{B}_{i,j-1}, \]
            and $\mathcal{B}_{i,j-1}$ has a red edge to $\{d_{i,j}\}$.
            Thus, the already contracted part of $R_i$ is connected to $A_{i,j}\cup B_{i,j}$ only through $\{d_{i,j}\}$.
            At $j=1$, this previously contracted path is absent.

            We now consider the contractions inside $A_{i,j}\cup B_{i,j}$.
            The only part that may have red degree $3$ is the part containing $d_{i,j}$.
            Every other part has red degree at most two.

            We first contract $\{a_{i,j}\}$ and $\{b_{i,j}\}$.
            The resulting part $\{a_{i,j}, b_{i,j}\}$ has at most one red neighbor $\{x_{i,j}\}$.
            This contraction creates no other red edge.
            The red edges among the already contracted parts remain unchanged.
            If $v_i \in S_j$, then $a_{i, j}$ and $b_{i, j}$ have the same neighbors in $G_\mathcal{C}$ and hence the resulting part $\{a_{i, j}, b_{i, j}\}$ has no red edge.
            We next contract $\{x_{i,j}\}$ and $\{z_{i,j}\}$.
            The resulting part $\{x_{i,j}, z_{i,j}\}$ has red edges to $\{a_{i,j}, b_{i,j}\}$ and $\{d_{i,j}\}$.
            After contracting $\{x_{i,j}, z_{i,j}\}$ and $\{y_{i,j}\}$, the resulting part $\{x_{i,j}, y_{i,j}, z_{i,j}\}$ has red edges to $\{a_{i,j}, b_{i,j}\}$ and $\{d_{i,j}\}$.
            Thus, every part has red degree at most two.

            We next contract $\{c_{i,j}\}$ and $\{d_{i,j}\}$.
            The resulting part $\{c_{i,j}, d_{i,j}\}$ has a red edge to $\{x_{i,j}, y_{i,j}, z_{i,j}\}$ and a red edge to $\{q_{i,j}\}$.
            If $j > 1$, it also has a red edge to $\mathcal{B}_{i,j-1}$.
            Thus, $\{c_{i,j},d_{i,j}\}$ has red degree at most $3$ while every other part has red degree at most two.
            At $j=1$, the red edge to $\mathcal{B}_{i,j-1}$ does not exist.
            Therefore, its red degree is at most two.

            We then contract $\{p_{i,j}\}$ and $\{q_{i,j}\}$.
            The resulting part $\{p_{i,j}, q_{i,j}\}$ has a red edge to $\{c_{i,j}, d_{i,j}\}$ and a red edge to $\{d_{i,j+1}\}$.
            Hence, its red degree is two.

            We next contract $\{a_{i,j},b_{i,j}\}$ and $\{c_{i,j},d_{i,j}\}$ into $\mathcal{A}_{i,j}$.
            Its red neighbors are the parts
            \[\{x_{i,j},y_{i,j},z_{i,j}\},\qquad\{p_{i,j},q_{i,j}\},\]
            and the part $\mathcal{B}_{i,j-1}$ if $j>1$.
            Thus, $\mathcal{A}_{i,j}$ has red degree at most $3$ while every other part has red degree at most two.

            Finally, we contract $\{x_{i,j}, y_{i,j}, z_{i,j}\}$ and $\{p_{i,j}, q_{i,j}\}$.
            The resulting part $\mathcal{B}_{i,j}$ has red edges to $\mathcal{A}_{i,j}$ and $\{d_{i,j+1}\}$.
            After all contractions inside $A_{i,j}\cup B_{i,j}$, every part has red degree at most two.
            At most one part has red degree three at any step of these contractions.

            For every $j\in[m]$, the construction gives $N_{G_{\mathcal{C}}}(h_j)\cap V(R_i)=B_{i,j}$.
            Thus, every part obtained by contracting vertices in $B_{i,j}$ is complete to $h_j$.
            Therefore, the corresponding edge is black.

            The neighbor lists above also show that the affected red component is acyclic.
            When present, the path formed by the already contracted parts is connected to the current blocks only through the part containing $d_{i,j}$.
            The two paths within these blocks run towards the part containing $a_{i,j},b_{i,j}$ and towards $q_{i,j},d_{i,j+1}$.
            These paths meet only at the part containing $d_{i,j}$.
            Hence, the component is a path or a subdivision of $K_{1,3}$.
            The other red components are unchanged.

            \proofparagraph{Step 2: Contractions across the row copies.}

            We restore the superscript $\ell$ on the parts obtained in each row copy.
            For every $i\in[n]$ and $\ell\in[k+1]$, put $\mathcal{A}_{i,m+1}^\ell=\{d_{i,m+1}^\ell\}$.
            After these contractions, every copy $R_i^\ell$ forms the red path
            \[\mathcal{A}_{i,1}^\ell-\mathcal{B}_{i,1}^\ell-\mathcal{A}_{i,2}^\ell-\mathcal{B}_{i,2}^\ell -\cdots-\mathcal{A}_{i,m}^\ell-\mathcal{B}_{i,m}^\ell-\mathcal{A}_{i,m+1}^\ell.\]

            We next contract the corresponding parts of these paths.
            Starting with $R_1^1$, we consider the remaining copies one at a time in the order
            \[R_1^2, \ldots, R_1^{k+1}, R_2^1, \ldots, R_2^{k+1}, \ldots, R_n^1, \ldots, R_n^{k+1}.\]
            When we consider a copy $R_i^\ell$, we contract its parts with the parts obtained from the previously considered copies in the order
            \[\mathcal{A}_{i,1}^\ell,\mathcal{B}_{i,1}^\ell,\mathcal{A}_{i,2}^\ell,\mathcal{B}_{i,2}^\ell, \ldots,\mathcal{A}_{i,m}^\ell,\mathcal{B}_{i,m}^\ell, \mathcal{A}_{i,m+1}^\ell.\]
            At every step, each connected component of the red graph is a path or a subdivision of $K_{1,3}$.
            Therefore, every spanning tree of such a component has at most $3$ leaves.

            \proofparagraph{Step 3: Remaining contractions.}

            For every $j\in[m]$, let $\widehat{B}_j$ be the part containing all sets $\mathcal{B}_{i,j}^\ell$ with $i\in[n]$ and $\ell\in[k+1]$.
            By construction, $\widehat{B}_j$ is complete to $h_j$.

            Finally, we contract $h_j$ with $\widehat{B}_j$ for every $j\in[m]$.
            The remaining red graph is a path.
            We contract this path from one endpoint.
            Therefore, every red graph in this contraction sequence has component max-leaf at most $3$.
            The contraction sequence can be constructed in polynomial time.
            This proves the claim.
        \end{proof}

        As shown above, $V$ has a set cover of size at most $k$ if and only if $G_{\mathcal{C}}$ has a dominating set of size at most $K$.
        Moreover, the described contraction sequence witnesses that $\reduced{\compmaxleaf}(G_{\mathcal{C}})\le 3$.
        Both $G_{\mathcal{C}}$ and the contraction sequence can be constructed in polynomial time.
        Thus, this gives a polynomial-time reduction from \textnormal{\textsc{Set Cover}} to \textnormal{\textsc{Dominating Set}} on graphs of reduced component max-leaf at most $3$, even when a witnessing contraction sequence is given.
        Since one can check in polynomial time whether a set of vertices is a dominating set of the required size, \textnormal{\textsc{Dominating Set}} is NP-complete on this class.

Every red graph in the sequence has maximum degree at most $3$, so the same sequence witnesses twin-width at most $3$.
Weak sparsity follows from \cref{ws:hardness-constructions}.
\end{proof}

\section{Weak sparsity of the hardness constructions}\label{sec:sparsity}\label{app:sparsity}

\begin{lemma}\label{ws:hardness-constructions}
The reductions for \textnormal{\textsc{$3$-Coloring}}, \textnormal{\textsc{Feedback Vertex Set}} and \textnormal{\textsc{Dominating Set}} construct $K_{3,3}$-subgraph-free graphs.
\end{lemma}
\begin{proof}
We use the notation of~\cref{sec:coloring,sec:fvs,sec:domination}.
In each construction, let $H$ be the independent set of vertices $h_j$.
We omit the row and copy indices.

We first consider \textsc{$3$-Coloring}.
The only pairs of vertices with at least three common neighbors are the pairs $\{w_k,z_k\}$ within a row.
For each such pair, we have
\[
 N(w_k)\cap N(z_k)=\{\beta_k,\gamma_{k+1},h_k\}.
\]
By construction, two vertices in the same $A$-block have at most one common neighbor, and two vertices in distinct $A$-blocks have at most two.
A vertex in an $A$-block and a vertex in a $B$-block have at most one common neighbor in the $A$-blocks.
They also have at most one common neighbor in the $B$-blocks, since the vertices in the $B$-blocks induce a matching.
Within a $B$-block, $w_k$ and $z_k$ have two common neighbors in the $A$-blocks, while every other pair in that block has no common neighbor in an $A$-block.
Two vertices in distinct $B$-blocks have at most one common neighbor in the $A$-blocks.
For every $h_j\in H$ and every other vertex $v$, we have $|N(h_j)\cap N(v)|\le2$.
Thus, the pairs with at least three common neighbors are pairwise disjoint.
Suppose that the graph contains a $K_{3,3}$ as a subgraph.
The three pairs of vertices on one side of this $K_{3,3}$ all have at least three common neighbors.
Any two of these pairs share a vertex, contradicting the pairwise disjointness.

For \textsc{Feedback Vertex Set}, we first delete all vertices in the $A$-blocks, each of which has degree two.
After these deletions, every vertex in a $B$-block has degree at most two: each of $a_j$ and $b_j$ has at most one neighbor in $H$, and $p_j$ and $q_j$ are adjacent only to each other and to $h_j$.
We next delete the vertices in the $B$-blocks and finally the vertices of $H$.
For \textsc{Dominating Set}, we delete the vertices in the following order:
\[
\begin{gathered}
 \{a_j,b_j,c_j,z_j:j\in[m]\},\quad \{y_j:j\in[m]\},\quad \{x_j:j\in[m]\},\\
 \{d_1,\ldots,d_{m+1}\},\quad \{p_j,q_j:j\in[m]\},\quad H.
\end{gathered}
\]
We delete the listed row vertices in all row copies in this order and then delete $H$.
In both cases, each vertex has degree at most two when we delete it.
Suppose that one of the graphs contains a $K_{3,3}$ as a subgraph.
The first vertex deleted from this $K_{3,3}$ has at least three remaining neighbors, a contradiction.
Therefore, both graphs are $K_{3,3}$-subgraph-free.
\end{proof}

\paragraph*{AI disclosure.}
The reduction used to prove~\cref{thm:mis} was suggested by GPT-6 Pro.
The authors simplified and improved the original argument establishing that the constructed instances exclude the $20 \times 20$ grid as an induced minor.

For \textsc{$3$-Coloring}, the authors developed the main structural idea.
To implement this idea, GPT-5.6 Sol suggested the source problem and gadgets for an initial reduction showing hardness on graphs of reduced component max-leaf at most~$58$.
The authors subsequently designed new gadgets using the same reduction idea, improving the bound to~$3$.
The \textsc{Feedback Vertex Set} and \textsc{Dominating Set} reductions and the write-up are due to the authors.

\label{page:last-main}
\newcommand{\etalchar}[1]{$^{#1}$}


\end{document}